\documentclass[11pt,onecolumn]{article}

\usepackage{times}
\usepackage{amsmath}
\usepackage{amssymb}
\usepackage{multirow}
\usepackage{xspace}
\usepackage{theorem}
\usepackage{graphicx}
\usepackage{ifpdf}
\usepackage{url}
\usepackage[bookmarks=false]{hyperref}
\usepackage{latexsym}
\usepackage{euscript}
\usepackage{xspace}
\usepackage{color}
\usepackage{makeidx}
\usepackage{picins,wrapfig}
\usepackage{xypic}
\usepackage{authblk}
\xyoption{all}

\long\def\remove#1{}

\newtheorem{theorem}{Theorem}[section] 
\newtheorem{lemma}[theorem]{Lemma}
\newtheorem{claim}[theorem]{Claim}

\newtheorem{cor}[theorem]{Corollary}

\newtheorem{definition}[theorem]{Definition}
\newenvironment{proof}{{\em Proof:}}{\hfill{\hfill\rule{2mm}{2mm}}}

\definecolor{darkred}{rgb}{1, 0.1, 0.3}

\newcommand{\myparagraph}[1]	{{\vspace*{0.08in}\noindent{\bf #1~}}}
\newcommand {\mm}[1] {\ifmmode{#1}\else{\mbox{\(#1\)}}\fi}
\newcommand{\denselist}{\itemsep 0pt\parsep=1pt\partopsep 0pt}
\newcommand{\eps}{{\varepsilon}}
\newcommand{\reals}	{{\rm I\!\hspace{-0.025em} R}}

\newcommand{\anotherG}		{{\tilde{G}}}
\newcommand{\Ghat}			{{\widehat{G}}}
\newcommand{\Jindex}			{\rho}
\newcommand{\athreshold}		{\tau}
\newcommand{\X}				{\mathcal{X}}
\newcommand{\aball}			{\mathrm{B}}

\newcommand{\realbad}			{really-bad\xspace}
\newcommand{\tG}				{{G^*}}
\newcommand{\tGr}				{{G_r^*}}
\newcommand{\oG}				{G}
\newcommand{\oE}				{E}
\newcommand{\tE}				{E^*}
\newcommand{\myER}			{{Erd\"{o}s-R\'{e}nyi\xspace}}
\newcommand{\dm}				{{L}}
\newcommand{\dtG}				{{d_{\tG}}}
\newcommand{\doG}				{{d_{\oG}}}
\newcommand{\myexp}			{\delta}
\newcommand{\mys}				{{\mathrm{s}}}
\newcommand{\Jthreshold}		{{\tau}}
\newcommand{\myprob}			{\mathrm{P}}
\newcommand{\length}				{{\mathrm{Len}}}

\newcommand{\mA}			{{\alpha}} 
\newcommand{\mB}			{{\beta}} 
\newcommand{\anotherGt}		{{\tilde{G}_{\athreshold}}}

\begin{document}
\title{A quest to unravel the metric structure behind perturbed networks}

\author[1]{Srinivasan Parthasarathy}
\author[2]{David Sivakoff}
\author[1]{Minghao Tian}
\author[1]{Yusu Wang}
\affil[1]{Computer Science and Engineering Dept., The Ohio State University, USA \\
\texttt{srini@cse.ohio-state.edu, tian.394@osu.edu, yusu@cse.ohio-state.edu}}
\affil[2]{Statistics and Mathematics Depts., The Ohio State University, USA \\
\texttt{dsivakoff@stat.osu.edu}}

\maketitle

\begin{abstract}
Graphs and network data are ubiquitous across a wide spectrum of scientific and application domains. 
Often in practice, an input graph can be considered as an observed snapshot of a (potentially continuous) hidden domain or process. 
Subsequent analysis, processing, and inferences are then performed on this observed graph. 
In this paper we advocate the perspective that an observed graph is often a noisy version of some discretized 1-skeleton of a hidden domain, and specifically we will consider the following natural network model: We assume that there is a true graph $\tG$ which is a certain proximity graph for points sampled from a hidden domain $\X$; while the observed graph $\oG$ is an \myER{} type perturbed version of $\tG$. 

Our network model is related to, and slightly generalizes, the much-celebrated small-world network model originally proposed by Watts and Strogatz.
However, the main question we aim to answer is orthogonal to the usual studies of network models (which often focuses on characterizing / predicting behaviors and properties of real-world networks). Specifically, we aim to recover the metric structure of $\tG$ (which reflects that of the hidden space $\X$ as we will show) from the observed graph $\oG$. 
Our main result is that a simple filtering process based on the \emph{Jaccard index} can recover this metric within a multiplicative factor of $2$ under our network model. 
Our work makes one step towards the general question of inferring structure of a hidden space from its observed noisy graph representation. 
In addition, our results also provide a theoretical understanding for Jaccard-Index-based denoising approaches. 
 \end{abstract}

\section{Introduction}
Graphs and networks are ubiquitous across a wide spectrum of scientific and application domains. Analyzing various types of graphs and network data play a fundamental role in modern data science. 
In the past several decades, there has been a large amount of research studying various aspects of graphs, ranging from developing efficient algorithms to process graphs, to information retrieval and inference based on graph data.  

In  many cases, we can view an input graph as an observed (discrete) 1-skeleton of a (potentially continuous) hidden domain or process. Subsequent analysis, processing, and inferences are then performed on this observed graph, with the ultimate goal being  
to understand the hidden space where the graph is sampled from. 
Many beautiful generative models for graphs have been proposed \cite{durrett2006random,penrose2003random}, aiming to understand this transition process from a hidden space to the observed 1-skeleton, and to facilitate further tasks performed on graphs. 

One line of such generative graph models assumes that an observed network is obtained by adding random perturbation to a specific type of underlying ``structured graph'' (such as a grid or a ring). 
For example, the much-celebrated small-world model by Watts and Strogatz \cite{watts1998collective} generates a graph by starting with a $k$-nearest neighbor graph spanned by nodes regularly distributed along a ring. It then randomly ``rewires'' some of the edges connecting neighboring points to instead connect nodes possibly far away. Watts and Strogatz showed that this simple model can generate networks that possess features of both a random graph and a proximity graph, and display two important characteristics often seen in real networks: low diameter in shortest path metric and high clustering coefficients. There have since been many variants of this model proposed so as to generate networks with different properties, such as adding random edges in a distance-dependent manner \cite{SW14, Kleinberg2000small}, or extending similar ideas to incorporate hierarchical structures in networks; e.g, \cite{Kleinberg02,WDN02}. 
There have also been numerous studies on characterizing statistical summaries, such as the average path lengths or the degree distributions, of small-world like networks; e.g \cite{bollobas1988diameter, faloutsos1999power}; see \cite{wang2003complex, bollobas2003mathematical} for a survey. 

\paragraph{Our work.} 
In this paper, we take the perspective that an observed graph can be viewed as a noisy snapshot of the discretized 1-skeleton of a hidden domain of interest, and propose the following network model: 
Assume that the hidden space that generates data is a ``nice'' measure $\mu$ supported on a compact metric space $\X = (X, d_X)$ (e.g, the uniform measure supported on an embedded smooth low-dimensional Riemannian manifold). Suppose that the data points $V$ are sampled i.i.d from this measure $\mu$, and the ``true graph'' $\tGr$ connecting them is the $r$-neighborhood graph spanned by $V$ (i.e, two points $u,v$ are connected if their distance $d_X(u, v) \le r$). 
The observed graph $\oG$  however is only a noisy version of the true proximity graph $\tGr$, and we model this noise by an Erd\"{o}s-R\'{e}nyi (ER) type perturbation -- each edge in the true graph $\tGr$ can be deleted with probability $p$, while a ``short-cut'' edge between two unconnected nodes $u, v$ could be inserted to $\oG$ with probability $q$.

To motivate this model, imagine in a social network a person typically makes friends with other persons that are close to herself in the unknown feature space modeled by our metric space $\X$. The distribution of people (graph nodes) is captured by the measure $\mu$ on $\X$. However, there are always (or may be even many) exceptions -- friends could be established by chance, and two seemingly similar persons (say, close geographically and in tastes) may not develop friendship. Thus it is reasonable to model an observed social network $\oG$ as an ER-type perturbation of the proximity graph $\tGr$ to account for such exceptions. 

The general question we hope to address is how to recover various properties of the hidden domain $\X$ from the observed graph $\oG$. In this paper we investigate a specific problem: how to recover the metric structure of $\tGr$ (induced by the shortest path distances in $\tGr$) from the noisy observation $\oG$. As we show in Theorem \ref{thm:metricapprox}, the metric structure of $\tGr$ ``approximates'' that of the hidden domain $\X$. 
Note that a few inserted ``short-cuts'' could significantly change the shortest path metric, one potential factor leading to the small-world phenomenon. 
Our main result is that a simple filtering procedure based on the so-called \emph{Jaccard index} can recover the shortest path metric of $\tGr$ within a multiplicative factor of $2$ (with high probabilities). We also provide some preliminary experimental results. 

\paragraph{Remarks and discussion.} 
The problem of recovering $\tGr$ from the observed graph $\oG$ is different and orthogonal to the usual studies on similar network models: Those studies often focus on characterizing the graphs generated by such models and whether those characteristics match with real networks. 
We instead aim to recover metric structure of a hidden true graph  $\tGr$ from a given graph $\oG$. There are different motivations for this task. For example, it could be that the true graph $\tGr$ is the real object of interest, and we wish to ``denoise'' the observed graph $\oG$ to get a more accurate representation of $\tGr$. Indeed, in \cite{goldberg2003assessing}, Godberg and Roth empirically show how to use small-world model to help remove false edges in protein-protein interaction (PPI) networks. See \cite{bhadauria2012efficient} for more examples. 

Furthermore, even if the observed graph $\oG$ is of interest itself, we may still want to recover information about the domain $\X$ where $\oG$ is generated from. For example, suppose we are given two networks $G_1$ and $G_2$ modeling say the collaboration networks from two different disciplines, and our goal is to compare the hidden collaboration structures behind the two disciplines. Comparing the precise graph structures of observed graphs $G_1$ and $G_2$ could be misleading, as even if they are generated from the same hidden space $\X$, they could still look different due to the random generation process. It is more robust if we can compare the two hidden spaces generating them instead. 

Finally, we remark that similar to the small-world network models, our model also overlays a random perturbation over a ``structured'' network. Indeed, our network model in some sense generalizes the small-world network model by Watts and Strogatz. Specifically, in the model by Watts and Strogatz (and some later variants), the  underlying ``structured'' network is a ring (or lattice). In our case, we assume that graph nodes $P$ are sampled from a measure $\mu$ and using the $r$-neighborhood proximity graph $\tGr$ to model this underlying ``structured'' network. 
This setup adds generality to our model: 
For example, it allows us to produce non-uniform and more complex degree distributions than those previously produced by starting with lattice vertices. At the same time, by putting conditions on the measure $\mu$, it still gives us sufficient structure to relate $\tGr$ and $\oG$, as we will show in this paper. We also point out that the theoretical results hold for graphs across a range of density, where the number of edges could range from $\Theta(n\log n)$ to $\Theta(n^2)$. 

\section{Model for Perturbed Network}
\label{sec:model}

We now introduce a general model to generate an observed network $\oG$. 
Suppose we are given a compact geodesic metric space $\X = (X, d_X)$ \footnote{A geodesic metric space is a metric space where any two points in it are connected by a path whose length equals the distance between them. Riemannian manifolds or path-connected compact sets in the Euclidean space are all geodesic metric spaces. 
}\cite{bridson2011metric}. 
Intuitively, we view an observed graph $\oG = (V, \oE)$ as a noisy 1-skeleton of $\X$, where graph nodes $V$ of $\oG$ are sampled from this hidden metric space. 
More precisely, we will assume that $V$ is sampled i.i.d. from a measure $\mu$ supported on $X$. 

\begin{definition}[Measure]\label{def:measure}
Given a topological space $X$, a \emph{measure} $\mu$ on $X$ is simply a function that maps every Borel subset $B$ of $X$ to a non-negative number $\mu(B)$, such that $\mu(\emptyset)=0$ and $\mu$ is $\sigma$-additive: that is the measure of a countable family of pairwise-disjoint Borel subsets of $X$ equals the sum of their respective measures. 
\end{definition}
In this paper, a measure is always a \emph{probability measure}, meaning that $\mu(X) = 1$. 
To provide sufficient structure to the observed graph $\oG$ so that it is not completely arbitrary, we want to assert some reasonable conditions on $\mu$. To this end, we consider doubling measures:   

\begin{definition}[Doubling measure \cite{heinonen2012lectures}]\label{def:doublingdim}
Given a metric space $\X = (X, d_X)$, let $\aball(x, r) \subset X$ denotes the open metric ball $\aball(x,r) = \{y\in X \mid d_X(x,y) < r\}$. 
A measure $\mu$ on $\X$ is said to be \emph{doubling} if balls have finite and positive measure and there is a constant $L=L(\mu)$ s.t. for all $x\in X$ and any $r > 0$, we have 
$\mu(\aball(x, 2r)) \le L \cdot \mu(\aball(x, r))$. 
We call $L$ the \emph{doubling constant} and say $\mu$ is an \emph{$L$-doubling measure}. 
\end{definition}

These conditions on the measure also implies conditions on the underlying space $X$ supporting the measure. Specifically, it is known that any metric space supporting a doubling measure has to be doubling as well, with its doubling constant depending on that of the measure \cite{heinonen2012lectures}. 

\paragraph{Network model.} 
We now describe our network model. 
Given a compact metric space $\X = (X, d_X)$ and an $\dm$-doubling measure $\mu$ supported on $X$, let $V$ be a set of $n$ points sampled i.i.d. from $\mu$. 
We assume that the \emph{true graph $\tGr = (V, \tE)$} is the $r$-neighborhood graph for some parameter $r > 0$; that is, $E(\tG) = \tE = \{ (u, v) \mid d_X(u, v) \le r, u, v \in V \}$. 
\begin{definition}\label{def:networkmodel}
The \emph{observed graph $\oG(r,p,q) = (V, \oE)$} is based on $\tGr=(V, \tE)$, but with the following two types of random perturbations: 
\begin{description}\denselist
\item {\sf $p$-deletion}: For each edge $(u,v) \in \tE$, $(u,v)$ is in the observed graph $\oG(r,p,q)$ with probability $1-p$ (that is, an edge in $\tE$ is deleted with probability $p$). 
\item {\sf $q$-insertion}: For any pair of nodes $u,v \in V$ s.t. $(u,v) \notin \tE$, we have that $(u,v) \in \oE$ with probability $q$. 
\end{description}
\end{definition}
Intuitively, in our model, the observed network $\oG$ is a random geometric graph sampled from the metric space $\X$ which then undergoes \myER{} type perturbation. In what follows, we often omit the parameters $r, p, q$ from the notations $\tGr$ and $\oG(r,p,q)$, when their choices are clear from the context. 
Note that both $\tG$ and $\oG$ are unweighted graphs (that is, all edges have weight $1$). 
We now equip each graph with its shortest path metric, and obtain two discrete metric space $(V, \dtG)$ and $(V, \doG)$ induced by $\tG$ and $\oG$, respectively. 

\paragraph{Problem statement and main results.} 
Adding short-cuts (via $q$-insertions) could significantly distort the shortest path metric in $\tG$. 
Our ultimate goal is to infer information about both $\X$ and $\mu$ where points are sampled from, through the study of the observed graph $\oG$. 
In this paper we aim to recover the metric structure of $\tG$ (as a reflection of metric structure of $\X$) from $\oG$. 
Specifically, we show that a simple filtering process based on the so-called Jaccard index can remove sufficient ``bad edges'' in $\oG$ so as to recover the shortest path metric of $\tG$ up to a factor of $2$ w.h.p. 

\begin{definition}[Jaccard index] \label{def:Jaccard}
Given an arbitrary graph $G$, let $N_{G}(u)$ denote the set of neighbors of $u$ in $G$ (i.e. nodes connected to $u \in V(G)$ by edges in $E(G)$). 
Given any edge $(u,v) \in E(G)$, the \emph{Jaccard index $\Jindex_{u,v}$} of this edge is defined as
\begin{align}\label{eqn:Jaccard}
 \Jindex_{u,v}(G) &= \frac{|N_{G}(u) \cap N_{G}(v)|}{|N_G(u) \cup N_G(v)|}. 
\end{align}
\end{definition}

We remark that Jaccard index is a popular way to measure similarity between a pair of nodes connected by an edge in a graph \cite{leicht2006vertex}, and has been commonly used in practice for denoising and sparsification purposes \cite{singer2013two,satuluri2011sigmod}. Our results provide a theoretical understanding for such empirical Jaccard-based denoising approaches.

The main result is stated in Theorem \ref{thm:maincombined}. To show how this is established, we show two results on the influence of the shortest path under the $p$-deletion (Theorem \ref{thm:deletiononly}) and under the $q$-insertion (Theorem \ref{thm:insertiononly}), respectively. 
The proof for Theorem \ref{thm:maincombined} combines the ideas for proofs of these two results. 

\myparagraph{Metric structures for $\tGr$ versus for $\X$.}
Our main results recover the shortest path metric for $\tGr$ approximately. 
In some sense, the metric of a proximity graph provides an approximation of that of $X$, the domain where input graph nodes are sampled from; 
see e.g, \cite{AL12,chazal2013persistence} for the case where $X$ is a smooth Riemannian manifold embedded in Euclidean space. 
We make this relationship precise for our setting as follows. 
The proof of this result is standard (see e.g, the proof of Theorem 5.2 of \cite{chazal2013persistence}). For completeness, we include the proof in Appendix \ref{appendix:thm:metricapprox}.  

\begin{theorem}\label{thm:metricapprox}
Let $(X, d_X)$ be a compact geodesic metric space and $\mu$ a doubling measure supported on $X$. Let $V_n$ be a set of $n$ points sampled i.i.d. from $\mu$, and $\tGr$ the $r$-neighborhood graph constructed on $V_n$ (each edge in $\tGr$ has equal weight $1$) with the associated shortest path metric $d_\tGr$. 
For any sample $V_n$, consider the distance between $r\cdot d_\tGr$ ($d_\tGr$ scaled by $r$) and $d_X$ restricted to the sample $V_n$; that is,  
$$\| r\cdot d_\tGr - d_X |_{V_n} \|_\infty := \max_{v, v' \in V_n} | r \cdot d_\tGr(v, v') - d_X (v, v') |. $$ 
Then we have that for a fixed $r$, $\limsup_{n \to \infty} \| r \cdot d_\tGr - d_X |_{V_n} \|_\infty \le r$ almost surely. 
\end{theorem}

\section{Recovering the shortest path metric of $\tG$}
\label{sec:metricrecovery}

To illustrate the main idea, we first consider the deletion-only and insertion-only perturbation of the true graph $\tG$ in Sections \ref{subsec:deletion} and \ref{subsec:insertion}, respectively. 
As we will see below, the main difficulty lies in handling insertions (short-cuts). 
We then combine the two cases and present our main result, Theorem \ref{thm:maincombined}. 
First, we describe one (natural) assumption on $r$ that we will use later in all our statements. 

Note that as $r$ tends to $0$, the corresponding $r$-neighborhood graph may be very sparse, and a sparse graph $\tGr$ is quite sensitive to random deletions and insertions. We would like to consider $r$ in a range where is meaningful. 
We make the  following assumption, asserting a lower-bound on the mass contained inside any metric ball of radius $r/2$: 
\vspace*{0.06in}\begin{description}
\item {\sf[Assumption-R]:} The parameter $r$ is large enough such that for any $x\in X$, $\mu(\aball(x, \frac{r}{2})) \ge \mys$ where $\mys$ satisfies  
$\mys \ge \frac{12 \ln n}{n-2} (= \Omega(\frac{\ln n}{n})). $ 
\end{description}
Intuitively, $r$ is large enough such that with high probability each vertex $v$ in $\tGr$ has degree $\Omega(\ln n)$.  
Note that requiring $r$ to be large enough to have an $\Omega(\ln n / n)$ lower bound on the measure of any metric ball is natural. For example, for a random geometric graph $G(r,n)$ constructed as the $r$-neighborhood graph for points sampled  i.i.d.\ from a uniform measure on a Euclidean cube, asymptotically this is the same requirement so that the resulting $r$-neighborhood graph is connected with high probability \cite{penrose2003random}. 

\begin{lemma}\label{lem:degreebound}
Under {\sf Assumption-R}, with probability at least $1 - n^{-5/3}$, all vertices in $\tGr$ have more than $\frac{s(n-1)}{3} > 4\ln n$ neighbors. 
\end{lemma}
\begin{proof}
For a fixed vertex $v\in V$, let $n_v$ be the number of points in $(V-\{v\}) \cap \aball(v, r)$. 
The expectation of $n_v$ is $ (n - 1) \cdot \mu(\aball(v,r)) \ge (n-1) \cdot \mu(B(v,\frac{r}{2}))\ge s(n-1)$. 
By Chernoff bounds, we thus have that 
\begin{align*}
\myprob[n_v < \frac{s(n-1)}{3}] &\le \myprob[n_v < \frac{1}{3}(n-1) \mu(\aball(v,r))]  \le e^{-\frac{(\frac{2}{3})^{2}}{2}(n-1)\mu(\aball(v,r))} \le e^{-\frac{8\ln n}{3}}  = n^{-\frac{8}{3}}
\nonumber\\ 
\end{align*}

It then follows from the union bound that the probability that all $n$ vertices in $V$ have degree larger than $s(n-1)/3$ is at least  $1 - n \cdot n^{-\frac{8}{3}} = 1 - n^{-5/3}$. 
\end{proof}

Since $\mu$ is a doubling measure, any two neighbors $(u,v)$ in the $r$-neighborhood graph $\tGr$ would share many neighbors. 
Specifically, if $(u,v)$ is an edge in $\tGr$, that is, $d_X(u,v) \le r$, then $\aball(u, r) \cap \aball(v, r)$ must contain a metric ball of radius $r/2$ (say centered at midpoint $z$ of a shortest path connecting $u$ to $v$ in $X$; see Figure \ref{fig:illustration} (a)). 
Thus by a similar argument as the proof of Lemma \ref{lem:degreebound}, we obtain the following bound on the number of common neighbors between the nodes $u, v$ if edge $(u,v) \in \tG$. 
\begin{cor}\label{lem:commonneighbors}
Assume that the graph nodes $V$ of $\tGr$ are sampled i.i.d from an $L$-doubling measure $\mu$ supported on a compact geodesic metric space $(X,d_X)$. Then under {\sf Assumption-R}, with probability at least $1 - n^{-2/3}$, any two neighbors $(u, v) \in \tGr$ have $\frac{s(n-1)}{3} > 4\ln n = \Omega(\ln n)$ number of common neighbors. 
\end{cor}

\subsection{Deletion only}
\label{subsec:deletion}

In this case, we assume that we will remove each edge in $\tG$ independently with probability $p$ to obtain an observed empirical graph $\Ghat$. 
Our goal is to relate the shortest path metrics $\dtG$ of $\tG$ and $d_{\Ghat}$ of $\Ghat$ respectively. 
Deletion-only means that shortest path distances in $\Ghat$ are larger than those in $\tG$. 
Furthermore, since any two nodes $u, v$ connected in $\tG$ share sufficient number ($\Omega(\ln n)$) of  common neighbors, intuitively, evan after removing a constant fraction of edges in $\tG$, we can still guarantee that w.h.p.\ $u$ and $v$ will have some common neighbors left, and thus $u$ and $v$ can be connected through that common neighbor by a path of length $2$ in $\Ghat$. 
Hence overall, w.h.p.\ the distortion in shortest path distance is at most a factor of $2$. 

\begin{definition}
Let $G$ and $G'$ be two graphs on the same set of nodes $V$, and equipped with graph shortest path metric $d_G$ and $d_{G'}$, respectively. 
By $d_G \le c d_{G'}$, we mean that for any two nodes $u, v\in V$, we have that $d_G(u,v) \le c d_{G'}(u, v)$. 
We say that $d_{G'}$ is \emph{a $c$-approximation of $d_G$} if $\frac{1}{c} d_G \le d_{G'} \le c d_G$. 
\end{definition} 

\begin{theorem}[Random deletion] \label{thm:deletiononly}
Let $V$ be $n$ points sampled i.i.d.\ from a probability measure $\mu: X \to \reals^+$ supported on a compact metric space $(X, d_X)$. 
Let $\tG$ be the $r$-neighborhood graph for $V$; and $\Ghat$ a graph obtained by removing each edge in $\tG$ independently with probability $p$. 
Under {\sf Assumption-R} and for $p < \frac{1}{2}e^{-\frac{9\ln n}{s(n-1)}}$, we have with probability at least $1 - \frac{1}{n^{\Omega(1)}}$, the shortest path metric $d_{\Ghat}$ is a $2$-approximation of the shortest path metric $\dtG$. 

Since $s > \frac{12\ln n}{n-1}$, the statement holds for $p < \frac{1}{2e^{3/4}}$. As $s$ becomes larger, the upper bound on $p$ gets closer to $1/2$. 
\end{theorem}

\begin{proof}
For a node $u \in V$, let $N_\tG(u)$ and $N_{\Ghat}(u)$ denote the set of neighbors 
of $u$ in graph $\tG$ and graph $\Ghat$, respectively. 

Since deletion cannot decrease the length of shortest paths, we have $d_{\tG} \le d_{\Ghat}$. We now show that $d_\Ghat \le 2 d_\tG$. 

Consider $(u,v) \in E(\tG)$. Assume that $u$ and $v$ share $k_{u,v}$ number of common neighbors; that is, $k_{u,v} = |N_\tG (u) \cap N_\tG(v)|$. 
The probability that $N_\Ghat(u) \cap N_\Ghat(v)  = \emptyset$ (i.e, $u$ and $v$ have no common neighbor in graph $\Ghat$) 
is thus at most $(2p)^{k_{u,v}}$.

On the other hand, by Corollary \ref{lem:commonneighbors}, with probability at least $1-n^{-2/3}$ we have that $k_{u,v} \ge s(n-1)/3$ for all $(u,v) \in E(\tG)$. 
Therefore, the probability that there exists $(u,v) \in E(\tG)$ with $N_{\Ghat}(u) \cap N_\Ghat(v) = \emptyset$ is at most
$n^{-2/3} + n^2 (2p)^{s(n-1)/3} < n^{-2/3} + n^2 (e^{-3\ln n}) < n^{-1/3}$, 
where we used the bound on $p$ to derive the first inequality. 

Hence with probability at least $1 - n^{-1/3}$, we have that for all edges $(u,v) \in E(\tG)$, their distance in $\Ghat$ satisfies $d_\Ghat(u,v) \le 2$ (via one of their common neighbor in $N_\Ghat(u) \cap N_\Ghat(v)$). 
This in turn implies that with probability at least $1 - n^{-1/3}$, for any path $\pi = \langle v_1, \ldots, v_m \rangle$ in $\tG$ with length $m$, we can find a path of length at most $2m$ in $\Ghat$ to connect $v_1$ to $v_m$ (as each edge $(v_i,v_{i+1})$ in $\pi$ corresponds to a path of length at most $2$ in $\Ghat$). If $u$ and $v$ are disconnected in $\tG$, then obviously they are still disconnected in $\Ghat$.
Hence, for any two $u, v\in V$, $d_\Ghat(u,v) \le 2 d_\tG(u,v)$, and the theorem follows. 
\end{proof}

\subsection{Insertion only} 
\label{subsec:insertion}

Now assume that the observed graph $\Ghat$ is generated from the true graph $\tG$ where all edges in $\tG$ also exist in $\Ghat$, and for any $u, v\in V$ with $(u,v) \notin E(\tG)$, we have $(u,v) \in E(\Ghat)$ with probability $q$. 
In this case, the shortest path metric can be significantly altered in $d_\Ghat$. 
Hence to recover the metric $d_\tG$, instead of operating on $\widehat G$ directly, we will construct \emph{another graph} $\anotherG$ from $\widehat G$ whose shortest path metric $d_{\anotherG}$ approximates $d_\tG$. 

We propose the following Jaccard-Index-based filtering process, which we call a $\athreshold$-Jaccard filtering, as it uses a parameter $\athreshold$. (Recall the definition of Jaccard index in Def. \ref{def:Jaccard}). We represent the output filtered (denoised) graph as $\anotherG_\athreshold$:  

\begin{description}\denselist
\item {\sf $\athreshold$-Jaccard filtering}: 
Given graph $\Ghat$, for each edge $(u,v) \in E(\Ghat)$, we insert the edge $(u,v)$ into $E(\anotherG_\athreshold)$ if and only if $\Jindex_{u,v}(\Ghat) \ge \athreshold$. That is, $V(\anotherG_\athreshold) = V(\Ghat)$ and $E(\anotherG_\athreshold) := \{ (u,v) \in E(\Ghat) \mid \Jindex_{u,v}(\Ghat) \ge \athreshold \}$. 
\end{description}

Below we first show that w.h.p., all ``good'' edges in the true $r$-neighborhood graph $\tG$ will have a large Jaccard index, so that they will be kept in $\anotherG_\athreshold$ after a $\athreshold$-Jaccard filtering procedure with appropriate $\athreshold$. We provide some discussions on the bounds of the parameters after the proof of this lemma. 

\begin{lemma}\label{lem:keepgoodedges}
Let $V$ be a set of $n$ points sampled i.i.d.\ from an $L$-doubling probability measure $\mu$ supported on a compact geodesic metric space $\X=(X, d_X)$. 
If {\sf Assumption-R} holds and $q \le c\mys$, then for $\forall \Jthreshold \leq \frac{1}{(6+\frac{1}{\ln{n}}+12c)L^{2}}$, we have with probability at least $1 - n^{-2/3}$, that $\Jindex_{u,v}(\Ghat) \ge \Jthreshold$ for \emph{all} pairs of nodes $u, v \in V$ with $d_X(u,v) \le r$.  

For example, if $c = \frac{1}{2}$ (i.e, $q \le \frac{\mys}{2}$), then the bound on $\Jindex_{u,v}$ holds for $\tau \le \frac{1}{13L^2}$. 
Note $c$ may not be a constant and can depend on $n$; as $c$ increases, the upper bound on $\tau$ decreases. 

\end{lemma}
\begin{proof}
Consider a fixed pair of nodes $u, v \in V$, and let $F = F(u,v)$ be the event that $d_X(u,v) \leq r$. Set $\mA_* = | N_\tG(u) \cap N_\tG(v)|$ to be the number of common neighbors of $u$ and $v$ in $\tG$. Let $\mB = |N_\Ghat(u) \cup N_\Ghat(v)|$ denote the total number of neighbors of $u$ and $v$ in the perturbed graph $\Ghat$. 

Since $\Ghat$ can have only more edges than $\tG$, $|N_\Ghat(u) \cap N_\Ghat(v)| \ge  | N_\tG(u) \cap N_\tG(v)| = \mA_*$ and thus $\Jindex_{u,v} (\Ghat) \ge \frac{\mA_*}{\mB}$.
In what follows, we prove that $\frac{\mA_*}{\mB} \ge \Jthreshold \cdot I_{F}$ (which implies that  $\Jindex_{u,v} (\Ghat) \ge \Jthreshold \cdot I_{F}$) with probability at least $1-2n^{-8/3}$. (Here, we use $I_A$ to denote the indicator random variable of the event $A$, and the conventions that $\Jindex_{u,v}(\Ghat)=0$ if $(u,v)\notin \Ghat$ and $0/0 = 0$.)

Note that $\mA_*$ is a random variable, which equals the number of (i.i.d.\ sampled) points from $V-\{u,v\}$ that fall in the region $\aball(u, r) \cap \aball(v,r)$. That is, conditional on $u$ and $v$, $\mA_*$ is drawn from a binomial distribution $Bin(n-2, p_{\mA_*})$ with $p_{\mA_*} = \mu(\aball(u, r) \cap \aball(v,r))$, and the conditional expectation of $\mA_*$ given $u$ and $v$ is $\myexp_{\mA_*} = (n-2) \cdot p_{\mA_*}$.

Now observe that, conditional on $u$ and $v$, the random variable $\mB-2$ (see footnote\footnote{The subtraction of $2$ in $\mB-2$ accounts for points $u$ and $v$, which are in $N_\Ghat(u) \cup N_\Ghat(v)$. Similarly, in the binomial distribution we will have only $n-2$, accounting for points in $V - \{u,v\}$.}) has distribution $Bin(n-2, p_{\mB})$ with $p_\mB = p_{\mB_*} + (1-p_{\mB_*})(2q-q^2)$, where $p_{\mB_*} = \mu(\aball(u,r) \cup \aball(v,r))$. 
Indeed, observe that, conditional on $u$ and $v$, points contributing to $\mB$ can be generated as follows. 
Let $U = \aball(u,r) \cup \aball(v,r)$. Independently, for each $i=1, \ldots, n-2$, we draw a point $x_i$ randomly from $\mu$ and we also perform an independent coin flip for this point, with probability of heads equal to $1-(1-q)^2 = 2q - q^2$. This quantity is the probability for a point {\it outside} $U$ to be connected to either $u$ or $v$ under edge-insertion probability $q$.
We set the indicator variable $y_i=1$ iff either $x_i\in U$, or $x_i\notin U$ and the $i^\text{th}$ coin flip is heads. 
Conditional on $u$ and $v$, the resulting $n-2$ indicator random variables $y_1, \ldots, y_{n-2}$ are i.i.d.\ with $\myprob[y_i = 1 \mid u,v] = p_{\mB_*} + (1-p_{\mB_*}) (2q-q^2) = p_\mB$. 
Therefore, given $u$ and $v$, the distribution of $\beta-2 = \sum y_i$ is $Bin(n-2, p_\mB)$. The conditional expectation of $\mB$ given $u$ and $v$, denoted $\myexp_\mB$, satisfies
\begin{equation}
(n-2) \cdot p_{\mB_*} \le \myexp_{\mB} = (n-2) \cdot p_\mB + 2 \le (n-2) \cdot p_{\mB_*} + (n-2) \cdot 2q + 2. \label{eqn:expB}
\end{equation}

Let us for now assume that $\frac{c_1 \myexp_{\mA_*}}{c_2 \myexp_{\mB}} \ge \Jthreshold I_{F}$ a.s.\ for constants $c_1 = 1 - \sigma_1$ and $c_2 = 1+\sigma_2$ with $0 < \sigma_1 < 1$ and $0< \sigma_2$ to be set shortly. 

If $d_X(u,v) \leq r$, then $\aball(u,r) \cap \aball(v,r)$ contains at least one metric ball of radius $r/2$ (say $\aball(z,r/2)$ with $z$ being the mid-point of a shortest path between $u$ and $v$ in $\X$; see Figure \ref{fig:illustration} (a)). 
\begin{figure}[t]
  \centering
  \begin{tabular}{ccc}
  \includegraphics[height=3cm]{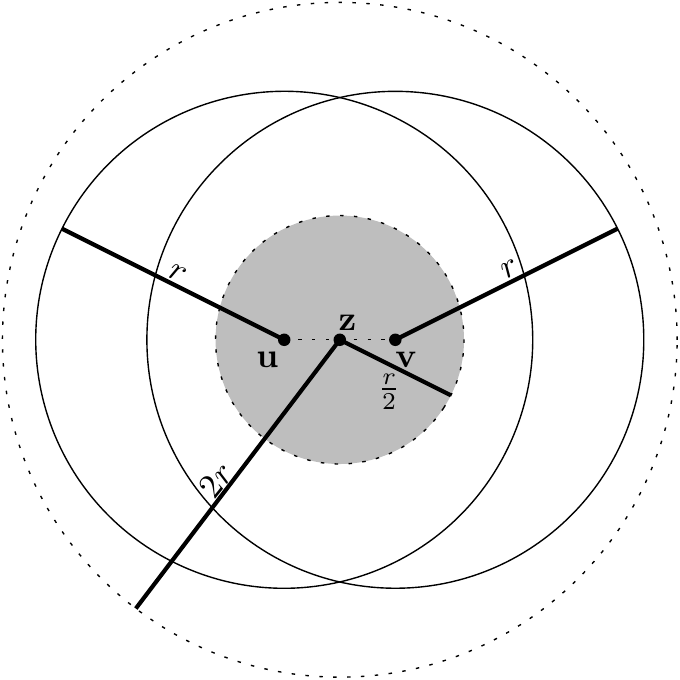} &\hspace*{0.3in} & \includegraphics[height=3cm]{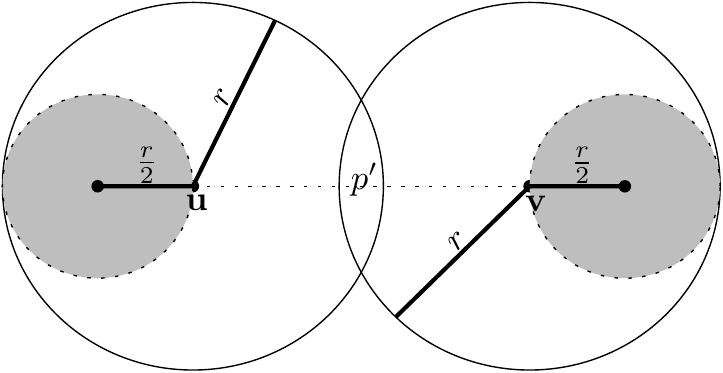} \\
(a) & & (b) 
  \end{tabular}
\vspace*{-0.15in}
  \caption{In these figures, we draw metric balls as Euclidean balls just for illustration purpose. (a) illustrates the bound $p_{\mA*} \ge \mu(\aball(z, r/2))$ which follows from $\aball(z,r/2) \subseteq \aball(u,r) \cap \aball(v,r)$. (b) Key observation for Lemma \ref{lem:removebadedges}: as $d_X(u,v) > r$, we have that the region $[B(u, r) \cup B(v, r)] \setminus [B(u,r) \cap B(v,r)]$ contains at least two metric balls, each of radius $r/2$.}
  \label{fig:illustration}
\end{figure}
Hence by {\sf Assumption-R}, on the event $d_X(u,v) \leq r$, we have 
$$\myexp_{\mA_*} \ge (n-2) \cdot \mu(\aball(z, r/2)) \ge (n-2) \cdot \frac{12 \ln n}{n-2} = 12 \ln n.$$ 
Similarly, using \eqref{eqn:expB}, the conditional expectation of $\mB$ satisfies 
\begin{align}\label{eqn:mB}
\myexp_{\mB} \ge (n-2) \cdot p_{\mB_*} &\ge (n-2) \cdot \mu(\aball(u,r)) \ge 12 \ln n.
\end{align}
We now set $\sigma_1 = 2/3$ and $\sigma_2 =1$. It then follows from Chernoff bounds that
\[
\myprob[ \mA_* < c_1 \myexp_{\mA_*}\mid u,v,F] + \myprob[\mB > c_2 \myexp_{\mB}\mid u,v] \le e^{-\frac{\sigma_{1}^{2}}{2}\myexp_{\mA_*}} + e^{-\frac{\sigma_{2}}{3}\myexp_{\mB}}~\le n^{-\frac{8}{3}} + n^{-4}.
\]
Taking expectation of the above with respect to $u$ and $v$ gives
\begin{equation}\label{eqn:newAB}
\myprob[ \mA_* < c_1 \myexp_{\mA_*} \mid F] + \myprob[\mB > c_2 \myexp_{\mB}] \le 2n^{-\frac{8}{3}}.
\end{equation}
On the other hand, since $\frac{\mA_*}{\mB}\ge 0$, we have
\begin{equation}\label{eqn:ratioA}
\begin{aligned}
\myprob[\frac{\mA_*}{\mB} < \Jthreshold I_{F}] &\le  \myprob[\frac{\mA_*}{\mB} < \Jthreshold \mid (\mA_* \ge c_1 \myexp_{\mA_*}) \wedge (\mB \le c_2 \myexp_{\mB}) \wedge F] \\ 
&\hspace{10pt} +  \myprob[(\{\mA_* < c_1 \myexp_{\mA_*}\} \vee \{\mB > c_2 \myexp_{\mB}\}) \wedge F]. 
\end{aligned}
\end{equation}
Since we assumed that $\frac{c_1 \myexp_{\mA_*}}{c_2 \myexp_{\mB}} \ge \Jthreshold I_{F}$, if $\mA_* \ge c_1 \myexp_{\mA_*}$ and $\mB \le c_2 \myexp_{\mB}$ and $d_X(u,v) \leq r$, then we have $\frac{\mA_*}{\mB} \geq \frac{c_1 \myexp_{\mA_*}}{c_2 \myexp_\mB} \ge \Jthreshold$. This means that  
$$\myprob[\frac{\mA_*}{\mB} < \Jthreshold \mid (\mA_* \ge c_1 \myexp_{\mA_*}) \wedge (\mB \le c_2 \myexp_{\mB})\wedge F] = 0. $$ 
Hence the first term in the right-hand side of \eqref{eqn:ratioA} is 0. 
Together with \eqref{eqn:newAB}, and recalling $\Jindex_{u,v} (\Ghat) \ge \frac{\mA_*}{\mB}$, we have
\begin{flalign*}
\myprob[\Jindex_{u,v}(\Ghat) &< \Jthreshold I_{F}] \le \myprob[\frac{\mA_*}{\mB} < \Jthreshold I_{F}] \le 2n^{-\frac{8}{3}}.
\end{flalign*}

By the union bound, the probability that $\Jindex_{u,v}(\Ghat) \ge \Jthreshold$ for all pairs of nodes $u,v \in V$ such that $d_X(u,v) \leq r$ is thus 
at least $1-\frac12 n^{2}(2n^{-\frac{8}{3}}) = 1 - n^{-\frac{2}{3}} $. 

Finally, we need to verify that $\frac{c_1 \myexp_{\mA_*}}{c_2 \myexp_{\mB}} = \frac{\myexp_{\mA_*}}{6\myexp_\mB} \ge \Jthreshold I_{F}$  holds for a.e.\ $u$ and $v$. This holds automatically if $d_X(u,v) > r$, so assume $d_X(u,v) \leq r$.
Recall that $\myexp_\mB \le (n-2)\cdot p_{\mB_*} + (n-2) \cdot 2q + 2$ by (\ref{eqn:expB}). 
Since $q \le c \mys$, we have $(n-2)2q \le 2(n-2)c \mys$.  
On the other hand, by {\sf Assumption-R}, $p_{\mB_*} \ge \mu(\aball(u, r)) \ge \mys$, hence $2(n-2)q \le 2(n-2)c \cdot p_{\mB_*}$. Combining this with the fact that  $(n-2)p_{\mB_*} \ge 12\ln{n}$ from \eqref{eqn:mB} (which also implies that $2 \le \frac{(n-2)p_{\mB_*}}{6\ln n}$), it then follows that
\begin{align}\label{eqn:finalA}
\frac{\myexp_{\mA_*}}{6\myexp_\mB} \ge \frac{\myexp_{\mA_*}}{6((n-2)(1+\frac{1}{6\ln{n}}) p_{\mB_*} + 2(n-2)c\cdot p_{\mB_*})} = \frac{p_{\mA_*}}{p_{\mB_*}}\cdot \frac{1}{6+\frac{1}{\ln{n}} + 12c}.
\end{align}


Now let $z$ be the midpoint of a geodesic connecting $u$ and $v$; see Figure \ref{fig:illustration} (a). Observe that $p_{\mA_*} \ge \mu(\aball(z, r/2))$, $p_{\mB_*} \le \mu(\aball(z, 2r))$ and since $\mu$ is $L$-doubling, we have: 
\begin{align}\label{eqn:finalB}
p_{\mB_*} \le \mu(\aball(z, 2r)) \le L \mu(\aball(z, r)) \le L^2 \mu(\aball(z,r/2)) \le L^2 p_{\mA_*}. 
\end{align}

Combining equations \eqref{eqn:finalA} and \eqref{eqn:finalB}, we have that if $\Jthreshold \le \frac{1}{(6 + \frac{1}{\ln{n}} + 12c)L^2}$, then $\frac{\myexp_{\mA_*}}{6\myexp_\mB} \ge \Jthreshold$ is satisfied. This proves the lemma. 
\end{proof}

\paragraph{Discussion on the bounds of parameters.} 
Lemma \ref{lem:keepgoodedges} implies that, with high probability, we will not remove any good edges if the doubling constant $L$ of the measure is at most $O(\frac{1}{\sqrt{\athreshold}})$ and the insertion probability is small ($q \le cs$). The requirement that $L = O(\frac{1}{\sqrt{\athreshold}})$ is rather mild; we now inspect the requirement $q \le cs$: Since $sn$ lower-bounds the degree of a node in the true graph $\tG$ (by Lemma \ref{lem:degreebound}), it is reasonable that the insertion probability $q$ is required to be small compared to $s$; as otherwise, the ``noise'' (inserted edges) will overwhelm the signal (original edges). 
Furthermore, it is important to note that $c$ is not necessarily a constant -- It can depend on $n$, but as $c$ increases, the upper bound of the admissible range for parameter $\tau$ decreases. 
\vbox{}

The following result complements Lemma \ref{lem:keepgoodedges} by stating that for insertion probability $q \le cs$, all ``really bad'' edges in $\Ghat$ will have small Jaccard index, and thus will be removed by our $\tau$-filtering process. 

In particular, we define an edge $(u, v) \in E(\Ghat)\setminus E(\tG)$ in the observed graph $\Ghat$ to be \emph{\realbad{}} if $N_{\tG}(u) \cap N_{\tG}(v) = \emptyset$. Note that $(u,v)\notin E(\tG)$ is equivalent to $d_{X}(u, v) > r$.

\begin{lemma}\label{lem:removebadedges}
Let $V$ be a set of $n$ points sampled i.i.d. from an $L$-doubling probability measure $\mu$ supported on a compact metric space $\X=(X, d_X)$. 
If {\sf Assumption-R} holds and $q \leq c\mys$, then for $\forall \Jthreshold \geq (c+2)q + 2(c+2)\sqrt{\frac{\ln{n}}{s(n-2)}}$, we have with probability at least $1 - n^{-2}$, $\Jindex_{u,v}(\Ghat) < \Jthreshold$ for \emph{all} pairs of nodes $u, v \in V$ such that $(u,v)$ is \realbad{}.  

For example, if $c = 1$ and $\mys \cdot n = \omega(\ln{n})$, then the condition on $\Jthreshold$ is that $\Jthreshold > 3q + o(1)$. 
\end{lemma}
\begin{proof}
Consider a fixed pair of nodes $u,v \in V$, and let $F = F(u,v)$ be the event that $N_{\tG}(u) \cap N_{\tG}(v) = \emptyset$ and $d_X(u,v) > r$.
Let $\mA = | N_\Ghat(u) \cap N_\Ghat(v) |$, 
\begin{align*}
\mA_{I} &= \big| \{x \in N_\tG(u) \cup N_\tG(v)\, :\, x~\text{is connected to both} ~u~\text{and}~ v~ \text{in}~\Ghat \}\big|,  ~\text{and}\\
 \mA_{o}&= \big|\{ x\notin N_\tG(u) \cup N_\tG(v)\, :\, x~ \text{is connected to both} ~u~ \text{and}~ v ~\text{in}~ \Ghat \}\big|. 
\end{align*}
Then we have $\mA = \mA_{I} + \mA_{o}$. 
Set $\mB_*=| N_\tG(u) \cup N_\tG(v)|$, so we have $|N_\Ghat(u) \cup N_\Ghat(v)| \geq \mB_{*} + \mA_{o} =:\mB$. 
It is easy to see that 
$$\Jindex_{u,v}(\Ghat) = \frac{\mA}{|N_\Ghat(u) \cup N_\Ghat(v)|}\leq \frac{\mA}{\mB_*+\mA_{o}} = \frac{\mA}{\mB}. $$ 
We aim to show that with high probability $\frac{\mA}{\mB} I_F< \tau$, which would then imply that $\Jindex_{u,v}(\Ghat) I_F< \tau$. 

Similar to the proof of Lemma \ref{lem:keepgoodedges}, we wish to understand the conditional distribution of random variables $\mA$ and $\mB$ given $u$, $v$ and $F$. 
A slight complication here is that it is possible that $B(u,r)\cap B(v,r) \neq \emptyset$. However, if $(u,v)$ is \realbad{}, we have $N_{\tG}(u) \cap N_{\tG}(v) = \emptyset$, meaning that there is no sample point from $V$ falls inside $B(u,r)\cap B(v,r)$ even if the region $B(u,r)\cap B(v,r) \neq \emptyset$. 
We claim that, conditional on the locations of $u$ and $v$ and the event $F$, the distribution of $\mA$ is $Bin(n-2, p_\mA)$ with $p_\mA = \frac{p_{\mB_*}-p'}{1-p'}q+\frac{1-p_{\mB_*}}{1-p'}q^2$ where $p_{\mB_*} = \mu(B(u,r)\cup B(v,r))$ and $p' =  \mu(B(u,r)\cap B(v,r))$. (Note, $p'$ is $0$ if $B(u,r)\cap B(v,r) = \emptyset$.) Indeed, we can imagine that, conditional on $u$, $v$ and $F$, points contributing to $\mA$ can be generated as follows: 

Let $U = [B(u,r) \cup B(v,r)] \setminus [B(u,r) \cap B(v,r)]$. Define the measure $\mu' : X \setminus [B(u,r)\cap B(v,r)] \to \reals$ to be the re-normalization of the probability distribution $\mu|_{X \setminus [B(u,r)\cap B(v,r)]}$ restricted to the domain outside $B(u,r) \cap B(v,r)$; that is, for any region $R \subset X\setminus [B(u,r)\cap B(v,r)]$, $\mu'(R) = \mu(R) / (1-p')$. 

We now draw a point $x_{i}$ randomly from the measure $\mu'$. 
The reason to exclude $B(u,r) \cap B(v,r)$ is because by our assumption, $N_{\tG}(u) \cap N_{\tG}(v) = \emptyset$; meaning $\mA$ is conditioned on $V \cap [B(u,r) \cap B(v,r)] = \emptyset$. 
We next flip two coins: For the first coin, the probability for head equals to $q$; while for the second one, the probability for head is $q^2$. 
We set the indicator variable $y_{i}=1$ if 
\begin{description}
\item (i) $x_{i}$ falls in $U$ and the first coin flip returns head, corresponding to the case where $x_i$ contributes to $\mA_I$, or
\item (ii) $x_{i}$ does not fall in $U$ but the second coin flip returns head, corresponding to the case where $x_i$ contributes to $\mA_o$. 
\end{description}
Conditional on $u$, $v$ and $F$, the resulting $n-2$ indicator random variables $y_{1}, \cdots, y_{n-2}$ are i.i.d.\ with $\myprob[y_{i}= 1] = \frac{p_{\mB_*}-p'}{1-p'}q+\frac{1-p_{\mB_*}}{1-p'}q^2 = p_\mA$. Therefore, given $u$, $v$ and $F$, the distribution of $\mA = \sum y_{i}$ is $Bin(n-2,p_\mA)$.

By a similar argument, we claim that the conditional distribution of $\mB$ is $Bin(n-2, p_\mB)$ with $p_\mB =  \frac{p_{\mB_*}-p'}{1-p'} + \frac{1-p_{\mB_*}}{1-p'}q^2.$

If $d_X(u,v) > r$, the region $[B(u, r) \cup B(v, r)] \setminus [B(u,r) \cap B(v,r)]$ contains at least two metric balls of radius $r/2$; see Figure \ref{fig:illustration} (b). Therefore, $p_{\mB_*} - p' \geq 2\mu (B(\frac{r}{2})) \geq 2s$. The conditional expectation of $\mA$ given $u$, $v$ and $F$, denoted by $\myexp_\mA (= (n-2)p_\mA)$, satisfies: 
\begin{align}
(n-2) \frac{p_{\mB_*} - p'}{1-p'} q \leq \myexp_{\mA} = (n-2) [\frac{p_{\mB_*} - p'}{1-p'} q + \frac{1}{1-p'} q^{2}] \le (c/2 + 1) (n-2) \frac{p_{\mB_*}-p'}{1-p'}  q   \label{eqn:mA-realbad}
\end{align}
where the last inequality uses that $q \leq cs \leq c \cdot \frac{p_{\mB_*} - p'}{2}$. 
The conditional expectation of $\mB$ given $u$, $v$ and $F$, denoted by $\myexp_\mB$, satisfies:
\begin{align}
\myexp_\mB = (n-2)p_\mB \geq (n-2) \frac{p_{\mB_*} - p'}{1-p'}. \label{eqn:mB-realbad}
\end{align}

Let us for now assume that $\frac{c_1 \myexp_{\mA}}{c_2 \myexp_{\mB}} I_F\leq \tau$ a.s. for $c_1 = 1 + \epsilon$ and some constant $c_2 = 1 - \sigma$ with $\epsilon = \frac{2}{q}\sqrt{\frac{\ln{n}}{s(n-2)}}$ and some $0 < \sigma < 1$ to be set later.

If $q \leq 2\sqrt{\frac{\ln{n}}{s(n-2)}}$, then we have $\epsilon \geq 1$.
In this case, combining Chernoff bounds with \eqref{eqn:mA-realbad} and the fact that $p_{\mB_*} - p' \geq 2\mu (B(\frac{r}{2})) \geq 2s$ obtained earlier, we have: 
\begin{align}\label{eqn:AandB21}
\myprob[\mA \ge (1+\epsilon) \myexp_{\mA} \mid u,v,F] \leq e^{-\frac{\epsilon}{3}\myexp_{\mA}} = e^{-\frac{2}{3q}\sqrt{\frac{\ln{n}}{s(n-2)}}\myexp_{\mA}} \leq e^{-\frac{2}{3q}\sqrt{\frac{\ln{n}}{s(n-2)}}(n-2)\frac{p_{\mB_*}-p'}{1-p'} q} \nonumber\\
\leq e^{-\frac{4}{3}\sqrt{(n-2)(\ln{n})s}} \leq e^{-\frac{4}{3}\sqrt{(n-2)(\ln{n})\frac{12\ln{n}}{n-2}}} = n^{-\frac{8\sqrt{3}}{3}}. 
\end{align}

Otherwise, we have $q > 2\sqrt{\frac{\ln{n}}{s(n-2)}}$, then $0 < \epsilon < 1$. In this case, by Chernoff bounds
\begin{align}\label{eqn:AandBp}
\myprob[\mA \ge (1+\epsilon) \myexp_{\mA} \mid u,v,F] \leq e^{-\frac{1}{2}\epsilon^{2}\myexp_{\mA}} &\leq e^{-2\frac{\ln{n}}{\mys (n-2)}\frac{1}{q^2}(n-2)\frac{p_{\mB_*}-p'}{1-p'} q} \nonumber\\
&= e^{-2\frac{\ln{n}(p_{\mB_*}-p')}{\mys q}} \leq e^{-2\ln{n}\cdot\frac{2\mys}{\mys q}} \leq n^{-4}. 
\end{align}

On the other hand, by Chernoff bounds, we have $\myprob[\mB \le c_2 \myexp_{\mB} \mid u,v,F] \le e^{-\frac{\sigma^2}{2}\myexp_{\mB}}.$ Note that $\myexp_{\mB} \geq (n-2)\cdot \frac{p_{\mB_*}-p'}{1-p'} \geq (n-2) \cdot 2s \geq 24 \ln{n} $. We now set $\sigma = 1/2$ so $c_2 = 1 - \sigma = 1/2$. By taking expectation with respect to $u$ and $v$, we have
\begin{align}\label{eqn:ABbound2}
\myprob[ \mA \ge c_1 \myexp_{\mA} \mid F] + \myprob[\mB \le c_2 \myexp_{\mB} \mid F] \leq 2n^{-4}.
\end{align}

Since $\tau>0$, we have that
\begin{equation}\label{eqn:ratioA2}
\begin{aligned}
\myprob[\frac{\mA}{\mB} I_F\ge \Jthreshold] 
&\le \myprob[\frac{\mA}{\mB} \ge \Jthreshold \mid (\mA < c_1 \myexp_{\mA}) \wedge (\mB > c_2 \myexp_{\mB})\wedge F] \\
& \hspace{15pt} + 
\myprob[\{(\mA \ge  c_1 \myexp_{\mA}) \vee (\mB \le c_2 \myexp_{\mB})\} \wedge F].
\end{aligned}
\end{equation}
Under our assumption that $\frac{c_1 \myexp_{\mA}}{c_2 \myexp_{\mB}} I_F \le \Jthreshold$ a.s., 
if $\mA < c_1 \myexp_{\mA}$, $\mB > c_2 \myexp_{\mB}$ and $d_X(u,v) > r$, then $\frac{\mA}{\mB}<\frac{c_1 \myexp_{\mA}}{c_2 \myexp_B} \le \Jthreshold$. Therefore, the first term on the right side of \eqref{eqn:ratioA2} 
is $\myprob[\frac{\mA}{\mB} \ge \Jthreshold \mid (\mA < c_1 \myexp_{\mA}) \wedge (\mB > c_2 \myexp_{\mB})\wedge F] = 0 $. It then follows from \eqref{eqn:ABbound2} that:  
\begin{align*}
\myprob[\frac{\mA}{\mB} I_F\ge \Jthreshold] \leq \myprob[(\mA \ge  c_1 \myexp_{\mA}) \vee (\mB \le c_2 \myexp_{\mB}) \mid F] \le 2n^{-4}
\end{align*}

Since $\Jindex_{u,v} (\Ghat) \le \frac{\mA}{\mB}$, we have 
$\myprob[\Jindex_{u,v}(\Ghat) I_F\ge \Jthreshold] \leq \myprob[\frac{\mA}{\mB} I_F\ge \Jthreshold]\le 2n^{-4} . $
By union bound, the probability that $\Jindex_{u,v}(\Ghat) < \Jthreshold$ for all pairs of nodes $u,v \in V$ satisfying the required conditions is thus 
at least $1-\frac12 n^{2}(2n^{-4}) = 1 - n^{-2} $. 

Finally, for the above argument to hold, we need to verify that $\frac{c_1 \myexp_{\mA}}{c_2 \myexp_{\mB}} I_F \le \Jthreshold $ holds for a.e.\ $u$ and $v$, where $c_1 = 1+\eps$ and $c_2 = 1/2$. This holds automatically if event $F$ doesn't happen, so assume $F$ happens. Recall that $\myexp_{\mA} \le (c/2 + 1)\cdot (n-2) \cdot \frac{p_{\mB_*}-p'}{1-p'} q$ and $\myexp_{\mB} \geq (n-2) \cdot \frac{p_{\mB_*}-p'}{1-p'}$ from \eqref{eqn:mA-realbad} and \eqref{eqn:mB-realbad}. 
This implies that: 
\begin{align}\label{eqn:removebadfinalA}
\frac{2\cdot \big(1+\frac{2}{q}\sqrt{\frac{\ln{n}}{\mys (n-2)}}\big) \cdot \myexp_{\mA}}{\myexp_{\mB}} &\leq \frac{2[(c/2 + 1) (n-2) \frac{p_{\mB_*}-p'}{1-p'} q + (c+2)\sqrt{\frac{\ln{n}}{\mys (n-2)}}(n-2) \frac{p_{\mB_*}-p'}{1-p'}]}{(n-2)\frac{p_{\mB_*}-p'}{1-p'}} \nonumber\\
 &= (c+2)q + 2(c+2)\sqrt{\frac{\ln{n}}{\mys (n-2)}}.
\end{align}

We then have that as long as $\Jthreshold \geq (c+2)q + 2(c+2)\sqrt{\frac{\ln{n}}{\mys (n-2)}}$, then $\frac{(1+\frac{2}{q}\sqrt{\frac{\ln{n}}{s(n-2)}})\myexp_{\mA}}{c_2 \myexp_{\mB}} \leq \Jthreshold$ is satisfied. The lemma then follows. 
\end{proof}

The above result implies that after Jaccard filtering, although there still may be some extra edges remaining in $\anotherG_\athreshold$, each such edge $(u,v)$ is not \realbad{}. In fact, $N_\tG(u) \cap N_\tG(v) \neq \emptyset$ for each such extra remaining edge $(u,v)$, implying that $d_\tG(u,v) \le 2$. 
This, combined with Lemma \ref{lem:keepgoodedges}, essentially leads to the following result. To simplify our statement, we assume $\mys n = \omega(\ln{n})$ in the following result; a more complicated form can be obtained without this assumption (similar to the statement in Lemma \ref{lem:removebadedges}). 

\begin{theorem}[Random Insertion]\label{thm:insertiononly}
Let $V$ be a set of $n$ points sampled i.i.d. from an $L$-doubling measure $\mu: X \to \reals^+$ supported on a compact metric space $(X, d_X)$. 
Let $\tG$ be the resulting $r$-neighborhood graph for $V$; and $\Ghat$ a graph obtained by inserting each edge not in $\tG$ independently with probability $q$. 
Let $\anotherGt$ be the graph after $\tau$-Jaccard filtering of $\Ghat$.  
Then, if ~{\sf Assumption-R} holds, $q \leq c\mys$ and $\mys n = \omega(\ln{n})$, then for $\forall \frac{1}{(6 + \frac{1}{\ln{n}} 12c)L^2} \geq \Jthreshold \geq (c+2)q + o(1)$, with high probability the shortest path distance metric $d_{\anotherG_\athreshold}$ satisfies: $\frac{1}{2} d_\tG \le d_\anotherGt \le d_\tG$; that is, $d_\anotherGt$ is a $2$-approximation for $d_\tG$ with high probability. 
\end{theorem}

\begin{proof}
Define $\mathcal{E}_1$ to be the event when all the edges in $\tG$ are present in $\anotherG_\athreshold$. 
By Lemma \ref{lem:keepgoodedges}, event $\mathcal{E}_1$ happens with probability at least $1 - n^{-2/3}$. Hence with at least this probability, $d_{\anotherG_\athreshold} \le d_\tG$. We now prove the lower bound for $d_\anotherGt$. 

Let $\mathcal{E}_2$ be the event where for all edges $(u, v) \in E(\anotherGt) \setminus E(\tG)$, $(u,v)$ is not \realbad{}. 
Lemma \ref{lem:removebadedges} says that event $\mathcal{E}_2$ happens with probability at least $1-n^{-2}$. 
To this end, observe that if an edge $(u,v)$ is not \realbad{}, then we have that $d_\tG(u,v) \le 2$ as $N_\tG(u) \cap N_\tG(v) \neq \emptyset$; specifically, there is a path $u \to w \to v$ connecting $u$ and $v$ through some $w\in N_\tG(u) \cap N_\tG(v)$.  

In what follows, assume both events $\mathcal{E}_1$ and $\mathcal{E}_2$ happen -- as discussed above, this assumption holds with high probability due to Lemmas \ref{lem:keepgoodedges} and  \ref{lem:removebadedges}. 

Now consider two points $u, v \in V$. First, suppose that $u, v$ are connected in $\anotherGt$.  
Let $\pi = \langle u_0 = u, u_1, \ldots, u_s = v\rangle$ be a shortest path between them in $\anotherGt$. Consider each edge $(u_i, u_{i+1})$ in the shortest path $\pi$ in $\anotherGt$. Either $(u_i, u_{i+1}) \in E(\tG)$, in which case we set $\hat \pi(u_i, u_{i+1}) = (u_i, u_{i+1})$. Otherwise if $(u_i, u_{i+1}) \notin E(\tG)$, then $(u_i,u_{i+1})$ is not \realbad{} due to event $\mathcal{E}_2$, meaning that $d_\tG(u_i, u_{i+1}) \le 2$. Hence we can find a path $\hat \pi (u_i, u_{i+1}) \subset \tG$ of length at most two to connect $u_i$ and $u_{i+1}$ in $\tG$. Putting these two together, we can construct a path $\hat \pi = \hat \pi(u_0, u_1) \circ \hat\pi(u_1, u_2) \circ \cdots \circ \hat\pi(u_{s-1}, u_s)$ connecting $u = u_0$ to $v = u_s$ in $\tG$. Clearly, this path has length at most $2s$. 
Hence, for any $u, v\in V$, we have that $d_\tG(u,v) \le 2 d_\anotherGt (u, v)$ if $(u,v)$ is connected in $\anotherGt$. 

If $u$ and $v$ are not connected in $\anotherGt$, then they are not connected in $\tG$ either; because if there is a path connecting them in $\tG$, then the same path is present in $\anotherGt$ as event $\mathcal{E}_1$ holds. 
Putting everything together, we then have that with high probability, for any $u, v\in V$, $d_\tG(u,v) \le 2 d_\anotherGt (u, v)$; that is $d_\anotherGt \geq \frac{1}{2}d_\tG $. The theorem then follows. 
\end{proof}

\section{Combined case} 
The arguments used in Sections \ref{subsec:deletion} and \ref{subsec:insertion} can be modified to prove our main result when the observed graph $\Ghat = G(r, p, q)$ is generated via the network model described in Definition \ref{def:networkmodel} that includes both edge deletion and insertion. 
Specifically, we now discuss the case where the perturbed graph $\Ghat = G(r, p, q)$  is generated via Definition \ref{def:networkmodel}. That is, it is an \myER-type perturbed version of $\tG$ where with probability $p$ an edge from $\tG$ is not present in the observed graph, while with probability $q$ an edge connecting points $u, v$ with $d_X(u,v) > r$ will be inserted into the observed graph. 
We still use $\anotherG_\athreshold$ to denote the graph after Jaccard-filtering with threshold $\tau$. 

First, given two graphs $G_1 = (V, E_1)$ and $G_2 = (V, E_2)$ spanning on the same set of vertices, we use $G_1 \cap G_2$ to denote the graph $(V, E_1 \cap E_2)$. 

\begin{lemma} \label{lem:combineddeletion} 
Let $V$ be a set of $n$ points sampled i.i.d. from an $L$-doubling probability measure $\mu$ supported on a compact metric space $\X=(X, d_X)$. Let $\tG$ be the $r$-neighborhood graph spanned by $V$, and $\Ghat$ the observed graph as defined in Definition \ref{def:networkmodel}. 
If ~{\sf Assumption-R} holds and $p < \frac{1}{2}e^{-\frac{9\ln n}{s(n-1)}}$, then with probability at least $1 - n^{-\Omega(1)}$, we have that the shortest path metric $d_{\Ghat\cap \tG}$ in the graph $\Ghat \cap \tG$ is bounded from above by $2 d_\tG$, implying $d_\Ghat \le d_{\Ghat\cap \tG} \le 2 d_\tG$. 
\end{lemma}

\begin{proof}
Since $\Ghat \cap \tG$ is a subgraph of $\Ghat$, we thus have $d_\Ghat \le d_{\Ghat\cap\tG}$.  
Now by an almost identical argument as the one used in the proof of Theorem \ref{thm:deletiononly}, we can prove that $d_{\Ghat \cap \tG} \le 2 d_\tG$ with high probability. Indeed, compared to the $\Ghat$ used in Theorem \ref{thm:deletiononly}, our $\Ghat \cap \tG$ also contain some randomly inserted edges, which can only further decrease the shortest path distances. 
The claim then follows. 
\end{proof}

\begin{lemma}\label{lem:combinedkeepgoodedges}

Let $V$ be a set of $n$ points sampled i.i.d. from an $L$-doubling probability measure $\mu$ supported on a compact metric space $\X=(X, d_X)$. If ~{\sf Assumption-R} holds, $p \leq \frac{1}{4}$ and $q \le \min \{\frac{1}{8},c\mys\}$, then for $\forall \Jthreshold \leq \frac{(1-p)^2}{(10+\frac{5}{3\ln{n}} +20c)L^{2}}$, 
we have with probability at least $1 - n^{-0.16}$, that $\Jindex_{u,v}(\Ghat) \ge \Jthreshold$ for \emph{all} pairs of nodes $u, v \in V$ with $d_X(u,v) \leq r$. 
\end{lemma}

\begin{proof}
Consider a fixed pair of nodes $u, v \in V$, and let $F = F(u,v)$ be the event that $d_X(u,v) \leq r$. Set $\mA_* = |(N_{\tG}(u)\cap N_{\tG}(v)) \cap (N_{\Ghat}(u) \cap N_{\Ghat}(v))|$, that is, the number of common neighbors of $u$ and $v$ in both $\tG$ and $\Ghat$; 
Let $\mB = |N_{\Ghat}(u) \cup N_{\Ghat}(v)|$ denote the total number of neighbors of $u$ and $v$ in the perturb graph $\Ghat$. It is easy to see that $\Jindex_{u,v} (\Ghat) = \frac{|N_\Ghat(u) \cap N_\Ghat(v)|}{|N_\Ghat(u) \cup N_\Ghat(v)|} \ge \frac{\mA_*}{\mB}$. 
In what follows, we will aim to prove that $\frac{\mA_*}{\mB} \ge \Jthreshold I_F$ with high probability. 

We claim that, conditional on the locations of $u$ and $v$, the distribution of $\mA_*$ is $Bin(n-2, \hat p_{\mA_{*}})$ with $\hat p_{\mA_{*}} = p_{\mA_*} \cdot (1-p)^2$, where $p_{\mA_*} = \mu(\aball(u, r) \cap \aball(v,r))$. Notice that, conditional on $u$ and $v$, the distribution of the number of common neighbors of $u$ and $v$ in $\tG$ is $Bin(n-2,p_{\mA_*})$ and the probability for each node to be still a common neighbor for both $u$ and $v$ in $\Ghat$ is $(1-p)^2$. Thus, by a similar (but more complicated) argument as in Lemma \ref{lem:keepgoodedges}, the conditional expectation of $\mA_*$ is 
$$\myexp_{\mA_*} = (n-2) \hat p_{\mA_*}= (n-2) \cdot p_{\mA_*} \cdot (1-p)^2. $$ and we also claim that, conditional on $u$ and $v$, the distribution of $\mB-2$ is $Bin(n-2, p_{\mB})$ with 
$p_\mB = p_{\mA_*} (1-p^2) + (p_{\mB_*} - p_{\mA_*})(1 - p + pq) + (1-p_{\mB_*})(1 - (1-q)^2)$, where $p_{\mB_*} = \mu(\aball(u,r) \cup \aball(v,r))$.

It is easy to see that $p_{\mB_*} \ge p_{\mA_*}$. And by the assumption on $p$ and $q$, we know $1-p-q > 0$. Therefore, the conditional expectation of $\mB$ given $u$ and $v$, denoted by $\myexp_\mB$, satisfies: 
\begin{align}\label{eqn:combinedmyexpB}
(n-2) \cdot [(1-p-q)(1-q)]p_{\mB_*} \le \myexp_{\mB} &= (n-2) p_\mB + 2 \nonumber\\
&\le (n-2) \cdot [(1-q)^2 - p^2]p_{\mB_*} + (n-2)(2q-q^2) + 2 \nonumber\\
&< (n-2) \cdot p_{\mB_*} + (n-2)2q + 2.
\end{align}

Let us for now assume that $\frac{c_1 \myexp_{\mA_*}}{c_2 \myexp_{\mB}} \ge \Jthreshold I_F$ for constants $c_1 = 1 - \sigma_1$ and $c_2 = 1+\sigma_2$ with $0 < \sigma_1 < 1$ and $\sigma_2 > 0$ to be set later. If $d_X(u,v) \leq r$, the region $B(u,r)\cap B(v,r)$ contains at least one metric ball of radius $r/2$ (recall Figure [\ref{fig:illustration}(a)]). Therefore, the conditional expectation of $\mA_*$ given $u$, $v$ and $F$, denoted by $\myexp_{\mA_*}$, satisfies: $$\myexp_{\mA_*} \ge (n-2) \cdot \mu(\aball(z, r/2))\cdot (1-p)^2 \ge 12 \frac{(n-2) \ln n}{n-2} \cdot (1-p)^2 = 12 \ln{n} \cdot (1-p)^2. $$ 
Similarly, using \eqref{eqn:combinedmyexpB}, the conditional expectation of $\mB$ given $u$ and $v$, denoted by $\myexp_{\mB_*}$, satisfies $\myexp_{\mB_*} \ge (n-2)(1-p-q)(1-q)p_{\mB_*} \ge 12 \ln n [(1-p-q)(1-q)]$. We now set $\sigma_1 = 4/5$ and $\sigma_2 = 1$. From the conditions stated in Lemma \ref{lem:combinedkeepgoodedges}, we have $p \leq 1/4$ and $q \le 1/8$. It then follows from Chernoff bounds that 
\begin{align}
\myprob[ \mA_* < c_1 \myexp_{\mA_*} \mid u,v,F] + \myprob[\mB > c_2 \myexp_{\mB} \mid u,v] \le e^{-\frac{\sigma_{1}^{2}}{2}\myexp_{\mA_*}} + e^{-\frac{\sigma_{2}}{3}\myexp_{\mB}} \le 2n^{-2.16}. \nonumber
\end{align}

Taking expectation of the above with respect to $u$ and $v$ gives
\begin{align}\label{eqn:combineAandB}
\myprob[ \mA_* < c_1 \myexp_{\mA_*} \mid F] + \myprob[\mB > c_2 \myexp_{\mB}] \le 2n^{-2.16}.
\end{align}

By a similar argument as used in the proof of Lemma \ref{lem:keepgoodedges}, we can derive: 
$$\myprob[\Jindex_{u,v}(\Ghat) < \Jthreshold I_F] \le \myprob[\frac{\mA_*}{\mB} < \Jthreshold I_F]\le 2n^{-2.16}. $$

By the union bound, the probability that $\Jindex_{u,v}(\Ghat) \ge \Jthreshold$ for all pairs of nodes $u,v \in V$ such that $d_X(u,v) \leq r$ is thus at least $1-\frac12 n^{2}(2n^{-2.16}) = 1 - n^{-0.16}$. 

What remains is to verify that $\frac{c_1 \myexp_{\mA_*}}{c_2 \myexp_{\mB}} = \frac{\myexp_{\mA_*}}{10\myexp_\mB} \ge \Jthreshold I_F$ holds for a.e.\ $u$ and $v$. This holds automatically if $d_X(u,v) > r$, so assume $d_X(u,v) \leq r$. Recall that $\myexp_\mB \le (n-2) \cdot p_{\mB_*} + (n-2)\cdot 2q + 2$. Since $q \le c \mys$, we have $(n-2) 2q \le 2(n-2)c \mys$. 
On the other hand, by {\sf Assumption-R}, $p_{\mB_*} \ge \mu(\aball(u, r)) \ge \mys$, hence $(n-2) 2q \le 2(n-2)c \cdot p_{\mB_*}$. Combining with $(n-2)p_{\mB_*} \ge s(n-2) \ge 12\ln n$, (which also implies that $2 \leq \frac{(n-2)p_{\mB_*}}{6\ln n}$), it then follows that
\begin{align}\label{eqn:combinefinalA}
\frac{\myexp_{\mA_*}}{10\myexp_\mB} \ge \frac{\myexp_{\mA_*}}{10[(n-2)(1+\frac{1}{6\ln n}) p_{\mB_*} + 2(n-2)c p_{\mB_*}]} = \frac{p_{\mA_*}}{p_{\mB_*}} \cdot \frac{(1-p)^2} {10 + \frac{5}{3\ln n} + 20c}.
\end{align}

Now let $z$ be the midpoint of the shortest path connecting $u$ and $v$ in $X$. Observe that $p_{\mA_*} \ge \mu(\aball(z, r/2))$, $p_{\mB_*} \le \mu(\aball(z, 2r))$ and since $\mu$ is $L$-doubling, it follows: 
\begin{align}\label{eqn:combinefinalB}
p_{\mB_*} \le \mu(\aball(z, 2r)) \le L \mu(\aball(z, r)) \le L^2 \mu(\aball(z,r/2)) \le L^2 p_{\mA_*}. 
\end{align}
Combing equations \eqref{eqn:combinefinalA} and \eqref{eqn:combinefinalB}, we have that as long as $\Jthreshold \leq \frac{(1-p)^2}{(10 + \frac{5}{3\ln n} + 20c)L^2}$, we have that $\frac{\myexp_{\mA_*}}{10\myexp_\mB} \geq \Jthreshold$ is satisfied. 
The lemma then follows. 
\end{proof}

Recall that we define an edge $(u, v) \in E(\Ghat)\setminus E(\tG)$ in the observed graph $\Ghat$ to be \emph{\realbad{}} if $N_{\tG}(u) \cap N_{\tG}(v) = \emptyset$.

\begin{lemma}\label{lem:combinedremovebadedges}
Let $V$ be a set of $n$ points sampled i.i.d. from an $L$-doubling probability measure $\mu$ supported on a compact metric space $\X=(X, d_X)$. 
Let $\tG$ and $\Ghat$ be the true graph and observed graph as described in Definition \ref{def:networkmodel}, respectively. 
If {\sf Assumption-R} holds and $p \leq \frac{1}{4}$ and $q \leq c\mys$, then for $\forall \Jthreshold \geq  \frac{(c+2)q}{1-p} + \frac{2(c+2)}{1-p}\sqrt{\frac{\ln{n}}{\mys (n-2)}}$, we have with probability at least $1 - n^{-1/4}$, that $\Jindex_{u,v}(\Ghat) < \Jthreshold$ for \emph{all} pairs of nodes $u,v \in V$ such that $(u,v)$ is \realbad{}.
\end{lemma}

\begin{proof}
Consider a fixed pair of nodes $(u,v) \in V$, and let $F = F(u,v)$ be the event that $N_{\tG}(u) \cap N_{\tG}(v) = \emptyset$ and $d_X(u,v) > r$. Let $\mA = |N_\Ghat(u)\cap N_\Ghat(v)|$, 
\begin{align*}
\mA_I &= \big| (N_\tG(u) \cup N_\tG(v)) \cap (N_\Ghat(u)\cap N_\Ghat(v)) \big|, ~\text{and} \\ 
\mA_o &= \big|\{x \not\in N_\tG(u) \cup N_\tG(v) | x ~\text{is connected to both} ~u~ \text{and}~ v~ \text{in}~ \Ghat \}\big|.
\end{align*} 
Obviously, $\mA = \mA_I + \mA_o$. 
Further set $\mB_* = \big| (N_\tG(u) \cup N_\tG(v))\cap (N_\Ghat(u) \cup N_\Ghat(v)) \big|$, then $|N_\Ghat(u) \cup N_\Ghat(v)| \geq \mB_* + \mA_{o}$. Setting $\mB := \mB_* + \mA_{o} -2$, we then have that $\Jindex_{u,v}(\Ghat) = \frac{\mA}{|N_\Ghat(u) \cup N_\Ghat(v)|} \le \frac{\mA}{\mB_* + \mA_{o} -2} = \frac{\mA}{\mB}$. We aim to show that with very high probability $\frac{\mA}{\mB} I_F < \tau$, which implies that $\Jindex_{u,v}(\Ghat) I_F < \tau$.

By applying the same technique as in Lemma \ref{lem:removebadedges}, we claim that, conditional on the locations of $u$ and $v$ and the event $F$, the distribution of $\mA$ is $Bin(n-2, p_\mA)$ with $p_\mA = \frac{p_{\mB_*} - p'}{1-p'} (1-p)q+ \frac{1-p_{\mB_*}}{1-p'} q^2$, where $p_{\mB_*} = \mu (B(u,r) \cup B(v,r))$ and $p'= \mu (B(u,r) \cap B(v,r))$. We also claim that the conditional distribution of $\mB$ given $u$, $v$ and $F$ is $Bin(n-2, p_\mB)$ with $p_{\mB} = \frac{p_{\mB_*} - p'}{1-p'} \cdot (1-p(1-q)) + \frac{1-p_{\mB_*}}{1-p'}q^2$ by a similar argument. 

If $d_X(u,v) > r$, we have $p_{\mB_*} - p' \geq 2\mu (B(\frac{r}{2})) \geq 2s$ (recall Figure [\ref{fig:illustration}(b)]). Therefore, the conditional expectation of $\mA$ given $u$, $v$ and $F$, denoted by $\myexp_\mA(=(n-2)p_\mA)$, satisfies:
\begin{align}\label{eqn:combinedmyexpA}
(n-2)\cdot \frac{p_{\mB_*} - p'}{1-p'}(1-p)q \leq \myexp_{\mA} &\leq (n-2) [\frac{p_{\mB_*} - p'}{1-p'} q + \frac{1}{1-p'} q^{2}]\nonumber\\
 &\le (c/2 + 1)\cdot (n-2) \cdot \frac{p_{\mB_*} - p'}{1-p'} \cdot q
\end{align}
where the last inequality follows from $q \le c\mys$. The conditional expectation of $\mB$ given $u$, $v$ and $F$, denoted by $\myexp_\mB(= (n-2)p_\mB)$, satisfies:
\begin{align}\label{eqn:combinedmyexpB2}
(n-2) \cdot \frac{p_{\mB_*} - p'}{1-p'} (1-p) \leq \myexp_{\mB} \leq (n-2) [\frac{p_{\mB_*} - p'}{1-p'} + \frac{1}{1-p'}q^2] + 2
\end{align}

Let us for now assume that $\frac{c_1 \myexp_{\mA}}{c_2 \myexp_{\mB}} I_F\leq \tau$ a.s.\ for $c_1 = 1 + \epsilon$ and some constant $c_2 = 1 - \sigma$ with $\epsilon = \frac{2}{q}\sqrt{\frac{\ln{n}}{s(n-2)}}$ and $0 < \sigma < 1$ to be set later. 

If $q \leq 2\sqrt{\frac{\ln{n}}{s(n-2)}}$, then we have $\epsilon \geq 1$. In this case, combining Chernoff bounds with \eqref{eqn:combinedmyexpA} and that $p_{\mB_*} - p' \geq 2\mys$ and $p \leq \frac{1}{4}$, we have:
\begin{align}\label{eqn:combineAandB21}
\myprob[\mA \geq (1+\epsilon) \myexp_{\mA}] \leq e^{-\frac{\epsilon}{3}\myexp_{\mA}} = e^{-\frac{2}{3}\sqrt{\frac{\ln{n}}{sn}}\frac{1}{q}\myexp_{\mA}} \leq e^{-\frac{2}{3}\sqrt{\frac{\ln{n}}{s(n-2)}}\frac{1}{q}(n-2)\frac{p_{\mB_*} - p'}{1-p'}(1-p)q} \nonumber\\
\leq e^{-\frac{4}{3}\sqrt{(n-2)(\ln{n})s}\frac{3}{4}} \leq e^{-\sqrt{(n-2)(\ln{n})\frac{12\ln{n}}{(n-2)}}} = n^{-2\sqrt{3}}
\end{align}

Otherwise, if $q > 2\sqrt{\frac{\ln{n}}{s(n-2)}}$, then $0 < \epsilon < 1$. In this case, by Chernoff bounds

\begin{align}\label{eqn:combineAandBp}
\myprob[\mA \geq (1+\epsilon) \myexp_{\mA}] \leq e^{-\frac{1}{2}\epsilon^{2}\myexp_{\mA}} &\leq e^{-2\frac{\ln{n}}{\mys (n-2)}\frac{1}{q^2}(n-2)\frac{p_{\mB_*} - p'}{1-p'}(1-p)q} \nonumber\\
&= e^{-2\frac{(\ln{n})(p_{\mB_*}-p')}{\mys q}\frac{3}{4}} \leq e^{-\frac{3}{2}(\ln{n})\frac{2\mys}{\mys q}} \leq n^{-3}
\end{align}

On the other hand, by Chernoff bounds, we have $\myprob[\mB \leq c_2 \myexp_{\mB} \mid u,v,F] \leq e^{-\frac{\sigma^2}{2}\myexp_{\mB}}$. Note that $\myexp_{\mB} \geq (n-2) \cdot \frac{p_{\mB_*} - p'}{1-p'} \cdot (1-p) \geq (n-2) \cdot \mu(B(u,r)\cup B(v,r)) \cdot \frac{3}{4}\geq 18 \ln{n} $. We now set $\sigma = 1/2$. By taking expectation with respect to $u$ and $v$, we have

\begin{align}\label{eqn:combineABbound2}
\myprob[\mA \geq c_1 \myexp_{\mA} \mid F] + \myprob[\mB \leq c_2 \myexp_{\mB} \mid F] \leq 2n^{-9/4}
\end{align}

By the same argument as used in the proof of Lemma \ref{lem:removebadedges}, we have 
$\myprob[\Jindex_{u,v}(\Ghat) I_F\ge \Jthreshold] \leq \myprob[\frac{\mA}{\mB} I_F\ge \Jthreshold]\le 2n^{-9/4}. $ By union bound, the probability that $\Jindex_{u,v}(\Ghat) < \Jthreshold$ for all pairs of nodes $u$, $v \in V$ satisfying the required conditions is thus at least $1-\frac12 n^2(2n^{-9/4}) = 1 -  n^{-1/4} $. 

Finally, note that for the above argument to hold, we need to verify that
$$\frac{c_1 \myexp_{\mA}}{c_2 \myexp_{\mB}} I_F= \frac{2((1+\frac{2}{q}\sqrt{\frac{\ln{n}}{s(n-2)}})\myexp_{\mA})}{\myexp_\mB} I_F \leq \Jthreshold$$ holds for a.e.\ $u$ and $v$. This holds automatically if $F$ doesn't happen, so assume $F$ happens. Recall that $\myexp_{\mA} \le (c/2 + 1)\cdot (n-2) \cdot \frac{p_{\mB_*} - p'}{1-p'} q$ by \eqref{eqn:combinedmyexpA} and $\myexp_{\mB} \geq (n-2)\cdot \frac{p_{\mB_*} - p'}{1-p'} \cdot (1-p)$ by \eqref{eqn:combinedmyexpB2}. This implies that 
\begin{align}\label{eqn:combineremovebadfinalA}
\frac{2((1+\frac{2}{q}\sqrt{\frac{\ln{n}}{\mys n}})\myexp_{\mA})}{\myexp_\mB} &\leq \frac{2((c/2 + 1)(n-2) \frac{p_{\mB_*} - p'}{1-p'} q + (c+2)\sqrt{\frac{\ln{n}}{\mys (n-2)}}(n-2) \frac{p_{\mB_*} - p'}{1-p'})}{(n-2) \frac{p_{\mB_*} - p'}{1-p'} (1-p)} \nonumber\\
&= \frac{c+2}{1-p}q + 2\frac{c+2}{1-p}\sqrt{\frac{\ln{n}}{\mys (n-2)}}.
\end{align}

We have that as long as $\Jthreshold \geq \frac{(c+2)q}{1-p} + \frac{2(c+2)}{1-p}\sqrt{\frac{\ln{n}}{\mys (n-2)}}$, then $\frac{c_1 \myexp_{\mA}}{c_2 \myexp_{\mB}} I_F\leq \Jthreshold$ is satisfied. The lemma then follows. 
\end{proof}

\begin{theorem}\label{thm:maincombined}
Let $V$ be a set of $n$ points sampled i.i.d. from an $L$-doubling measure $\mu: X \to \reals^+$ supported on a compact metric space $(X, d_X)$. 
Let $\tG$ be the resulting $r$-neighborhood graph for $V$; and $\Ghat$ a graph obtained by the network model $G(r,p,q)$ described in Definition \ref{def:networkmodel}. 
Let $\anotherGt$ be the graph after $\tau$-Jaccard filtering of $\Ghat$.  
Then, if ~{\sf Assumption-R} holds, $p \leq \frac{1}{4}$, $q \leq \min \{\frac{1}{8},c\mys\}$ and $\mys n = \omega(\ln{n})$, then for any $\tau$ such that $\frac{(1-p)^2}{(10 + \frac{5}{3\ln{n}} + 20c)L^2} > \Jthreshold > \frac{(c+2)q}{1-p} + o(1)$, with high probability the shortest path distance metric $d_{\anotherG_\athreshold}$ is a 2-approximation of the shortest path metric $d_\tG$ of the true graph $\tG$. 
\end{theorem}

\begin{proof}
Let $\mathcal{E}_1$ denote the event where $d_{\Ghat\cap \tG} \le 2d_\tG$. By Lemma \ref{lem:combineddeletion}, event $\mathcal{E}_1$ happens with probability at least $1 - n^{-\Omega(1)}$. 

Let $\mathcal{E}_2$ denote the event where all edges $\Ghat \cap \tG$ are also contained in the edge set of the filtered graph $\anotherGt$; that is, $\Ghat \cap \tG \subseteq \anotherGt$. 
By Lemma \ref{lem:combinedkeepgoodedges}, event $\mathcal{E}_2$ happens with probability at least $1 - n^{-0.16}$.
It then follows that:   
\begin{align}
\text{If both events } \mathcal E_1 \text{~and~} \mathcal E_2 \text{~happen, then~} d_\anotherGt \le d_{\Ghat \cap \tG} \le 2 d_\tG. \nonumber
\end{align}

What remains is to show $d_\tG \le 2 d_\anotherGt$. 
To this end, we define $\mathcal E_3$ to be the event where for all \realbad{} edges $(u,v)$ in $\Ghat$, we have $\Jindex_{u,v}(\Ghat) < \tau$. If $\mathcal E_3$ happens, then it implies that for an arbitrary edge $(u, v) \in E(\anotherGt)$, either $(u, v) \in E(\tG)$ or $d_\tG(u, v) = 2$ (since $N_\tG(u) \cap N_\tG(v) \neq \emptyset$).
By Lemma \ref{lem:combinedremovebadedges}, event $\mathcal E_3$ happens with probability at least $1-n^{-1/4}$. 

By union bound, we know that $\mathcal E_1$, $\mathcal E_2$ and $\mathcal E_3$ happen simultaneously with high probability. 

Using the same argument as in the proof of Theorem \ref{thm:insertiononly}, it then follows that given any $u, v\in V$ connected in $\anotherGt$, we can find a path in $\tG$ of at most $2 d_\anotherGt(u, v)$ number of edges to connect $u$ and $v$. 
Furthermore, event $\mathcal E_1$ implies that if $u$ and $v$ are not connected in $\anotherGt$, then they cannot be connected in $\tG$ either. 
Putting everything together, we thus obtain $d_\anotherGt \geq \frac{1}{2}d_\tG$. Theorem \ref{thm:maincombined} then follows. 
\end{proof}

\myparagraph{Extension to local doubling measure.}
We can relax the $L$-doubling condition of the measure $\mu$ where points are sampled from to a \emph{local doubling} condition, where the $L$-doubling property is only required to hold for metric balls of small radius. Specifically, 

\begin{definition}[$(R_0, L_{R_0})$-doubling measure]\label{def:localdoublingdim}
Given a metric space $\X = (X, d_X)$, a measure $\mu$ on $\X$ is said to be \emph{$(R_0, L_{R_0})$-doubling} if balls have finite and positive measure and there is a constant $L_{R_0}$ s.t. for all $x\in X$ and any $0 <R \leq R_0 $, we have 
$\mu(\aball(x, 2R)) \le L_{R_{0}} \cdot \mu(\aball(x, R))$. 
\end{definition}
All our results hold for $(R_0, L_{R_0})$-doubling measure, as long as the parameter $r$ generating the true graph $\tG_r$ satisfies $r < R_0$. 
The proofs follow the same argument as those for $L$-doubling measure almost verbatim, and thus are omitted.

\section{Some empirical results}
\label{sec:exp}

We provide some proof-of-principle experimental results to show the effectiveness of the Jaccard filtering process. We report on two sets of experiments --- one controlled experiment on synthetic datasets and the other on real-world network datasets.

\begin{figure}[htbp]
\centering
\begin{tabular}{cc}
\includegraphics[width = .35\textwidth]{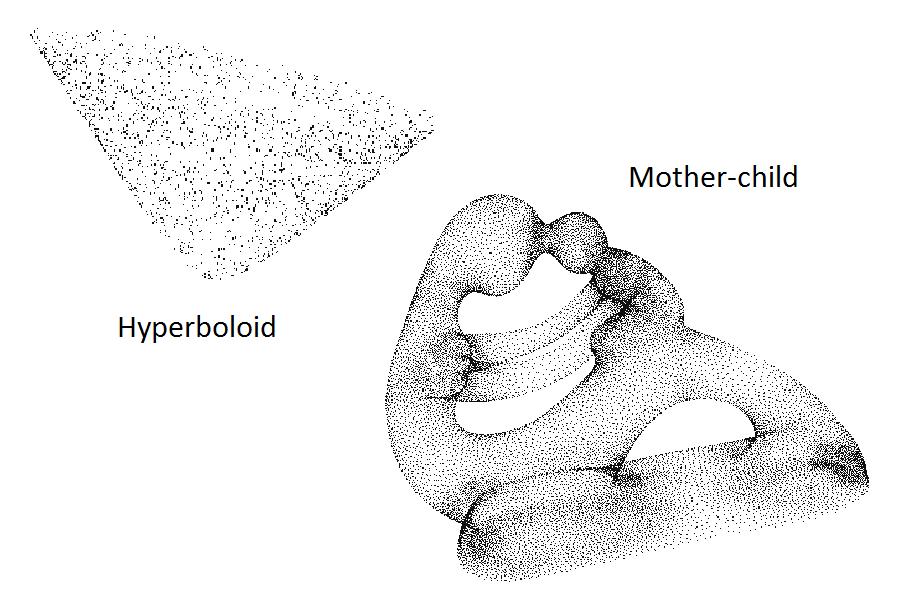}& 
\includegraphics[width= .35\textwidth]{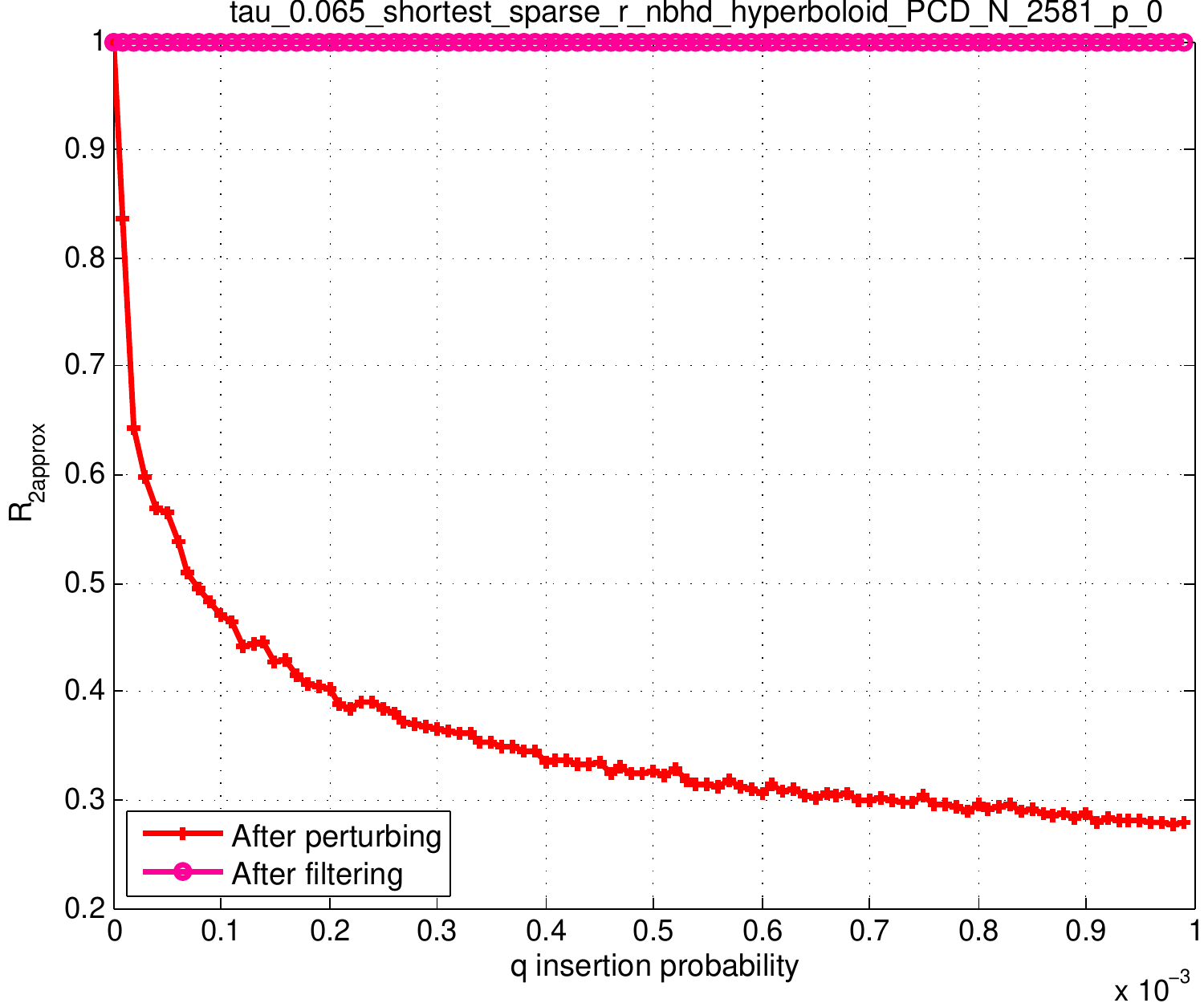}\\
(a) & ~~~(b)
\end{tabular}

\medskip
 
\includegraphics[width= .35\textwidth]{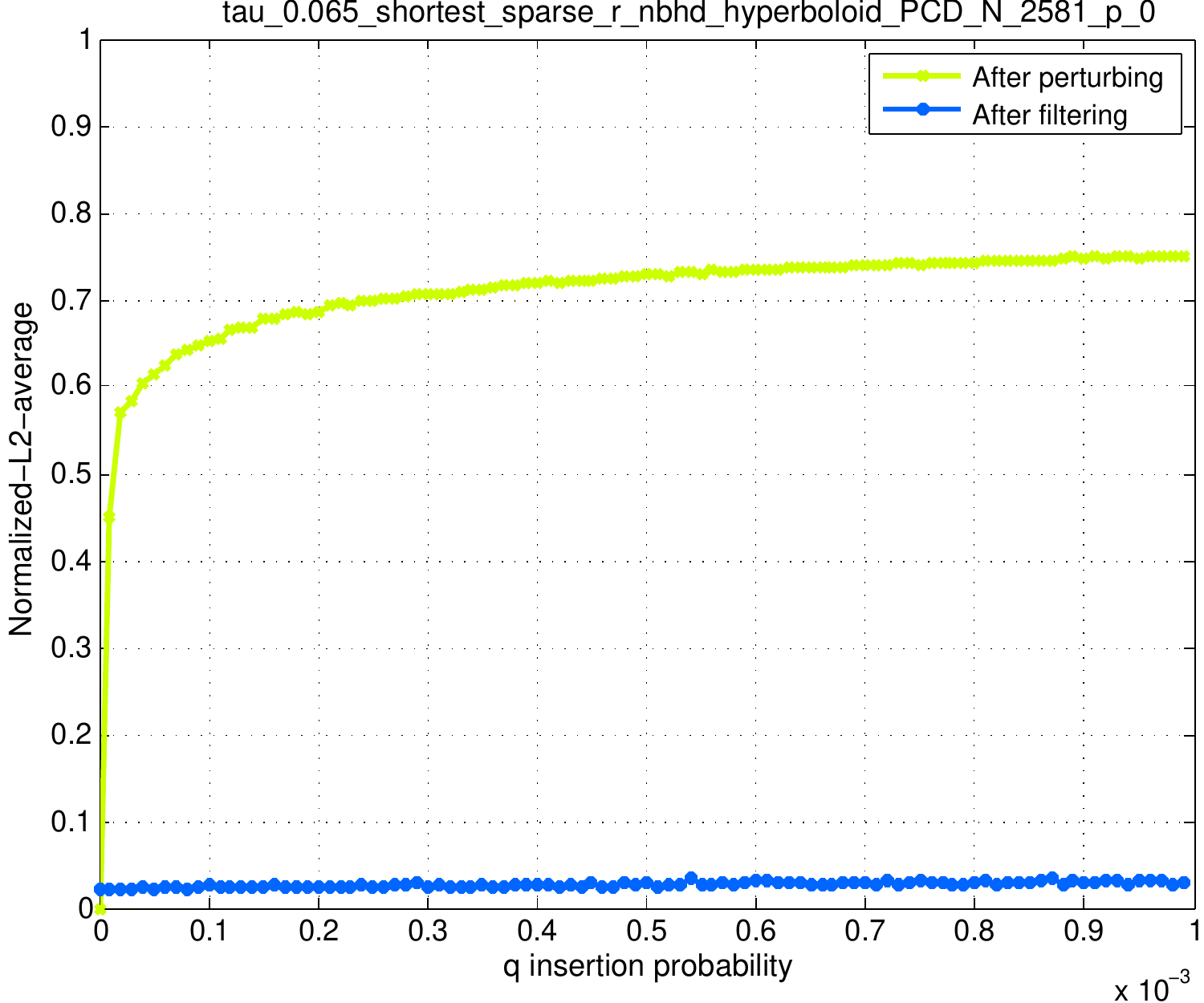}\\
~~~(c)
\caption{(a) $2.5K$ points sampled from a hyperboloid surface and $24K$ points sampled from mother-child model. 
(b) Comparison of 2-approximation rate $R_{2approx}$ as insertion probability (x-axis) increases. Top curve is after Jaccard-filtering, while bottom one is for perturbed graph without filtering. (c) Normalized $L_2$-average error with top curve being the one without filtering, and the bottom one (with significantly lower error) for after Jaccard-filtering. \label{fig:hyperboloidsparse}}
\end{figure}
\myparagraph{Synthetic datasets with ground truth.} 
In this experiment we seek to demonstrate that the Jaccard filtering approach works in a robust manner as predicted by our theoretical results. In particular, we start with the following two measures: 
$\mu_1$ is the ``quasi-uniform'' measure on the hyperboloid $S_1$ specified by $x^2+y^2-z^2 = 1$ \cite{Asta};  
$\mu_2$ is a non-uniform measure on the mother-and-child geometric model $S_2$ (see Figure \ref{fig:hyperboloidsparse}), where the measure is proportional to the local feature size at each point. 
For each $\mu_i$, we sample $n$ points $V$ i.i.d and build an $r$-neighborhood graph (we will specify choice of $r$ later). 
See Figure \ref{fig:hyperboloidsparse} (a) for illustration of input samples. 
This gives rise to a ground-truth neighborhood graph $G^*_r$. 
We next generate a set of observed graph $\oG_{p,q}$, varying the deletion probability ($p$) and insertion probability ($q$). 
Using a fixed parameter $\tau$, we perform $\tau$-Jaccard filtering for each $\oG_{p,q}$ to obtain a filtered graph $\widehat{G}^\tau_{p,q}$. 

To measure the difference between two metrics (represented as matrices) $D$ and $D'$, we use two types of error to be introduced shortly. 
But first, note that since we delete edges, the connectivity of the graph may change. Assume that $D_{i,j} = \infty$ if the two corresponding points $p_i$ and $p_j$ are not connected in the graph. Note that if $D_{i,j} = \infty$ and $D'_{i,j} = \infty$, the relation $\frac{1}{2}D_{i,j} \le D'_{i,j} \le 2D_{i,j}$ still holds.
\begin{itemize}
\item {\bf $2$-approximation rate $R_{2approx}$} is defined by 
$$R_{2approx}(D, D') = \frac{| \{ (i,j), 1\le i < j \le n \mid \frac{1}{2}D_{i,j} \le D'_{i,j} \le 2D_{i,j} \}|}{n(n-1)/2}. $$ 
In other words, $R_{2approx}$ is the ratio of ``good'' pairwise distances from $D'$ that $2$-approximate those in $D$. 
\end{itemize}

To analyze $L_2$-type error, we need to avoid the cases that $D_{i,j}$ is not comparable with $D'_{i,j}$. Thus, we collect the following \emph{good-index set}\\ 
{\small $I_{good}(D, D') = \{ (i, j), 1\le i < j \le n \mid$ either $(D_{i,j} < \infty) \wedge (D'_{i,j}< \infty)$; or $(D_{i,j} = \infty)\wedge(D'_{i,j} = \infty) \}$.} 
\begin{itemize}
\item {\bf Normalized $L_2$-average error $\delta_N(D, D')$.} 
First, we define root-mean-squared (RMS) error by 
$$\delta (D, D') = \sqrt{ \frac{\sum_{(i,j) \in I_{good}} (D_{i,j}-D'_{i,j})^2}{|I_{good}|}} $$
where note that if $D_{i,j} = \infty$ and $D'_{i,j} = \infty$, then $D_{i,j} - D'_{i,j} = 0$.
We then normalize it by the normalized $L_2$ norm of $D$; that is, 
$$ \delta_N(D, D') = \frac{\delta(D, D')}{\sqrt{\frac{1}{|\{ i< j, D_{i,j} < \infty\}|}\sum_{i< j, D_{i,j}<\infty} D^2_{i,j}}}. $$ 
\end{itemize}

Let $D_G$ denote the shortest path metric induced by a graph $G$. 
We compare the $2$-approximation rate $R_{2approx}(D_\tG, D_{G_q})$ for the sequence of observed graphs $G_{q}$ for increasing insertion probability $q$ ($x$-axis in all the plots) with $R_{2approx}(D_\tG, D_{\widehat{G}^\tau_q})$s for the sequence of filtered graph $\widehat{G}_q^\tau$; 
while we also compare the normalized $L_2$ error $\delta_N(D_\tG, D_{G_q})$ versus $\delta_N( D_\tG, D_{\widehat{G}^\tau_q})$s for increasing $q$s. 

In the following experiments, we choose $r$ (to build the $r$-neighborhood graph) to be (a) (``sparse'') twice or (b) (``dense'') five times of the average distance from a point to its 10-th nearest neighbor in $P$. 
\begin{figure}[htbp]
\centering
\begin{tabular}{cc} 
\includegraphics[width= .35\textwidth]{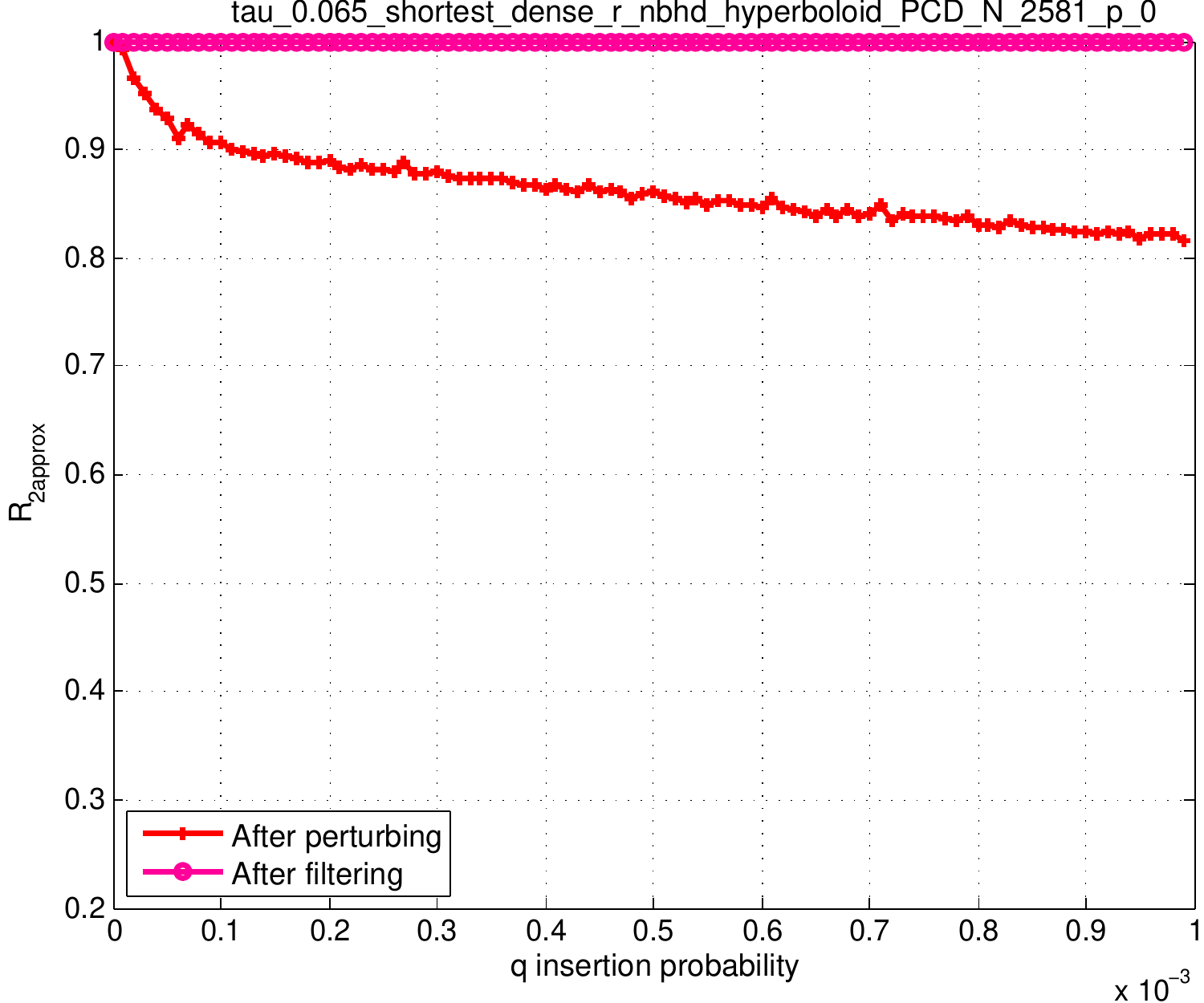} & 
\includegraphics[width= .35\textwidth]{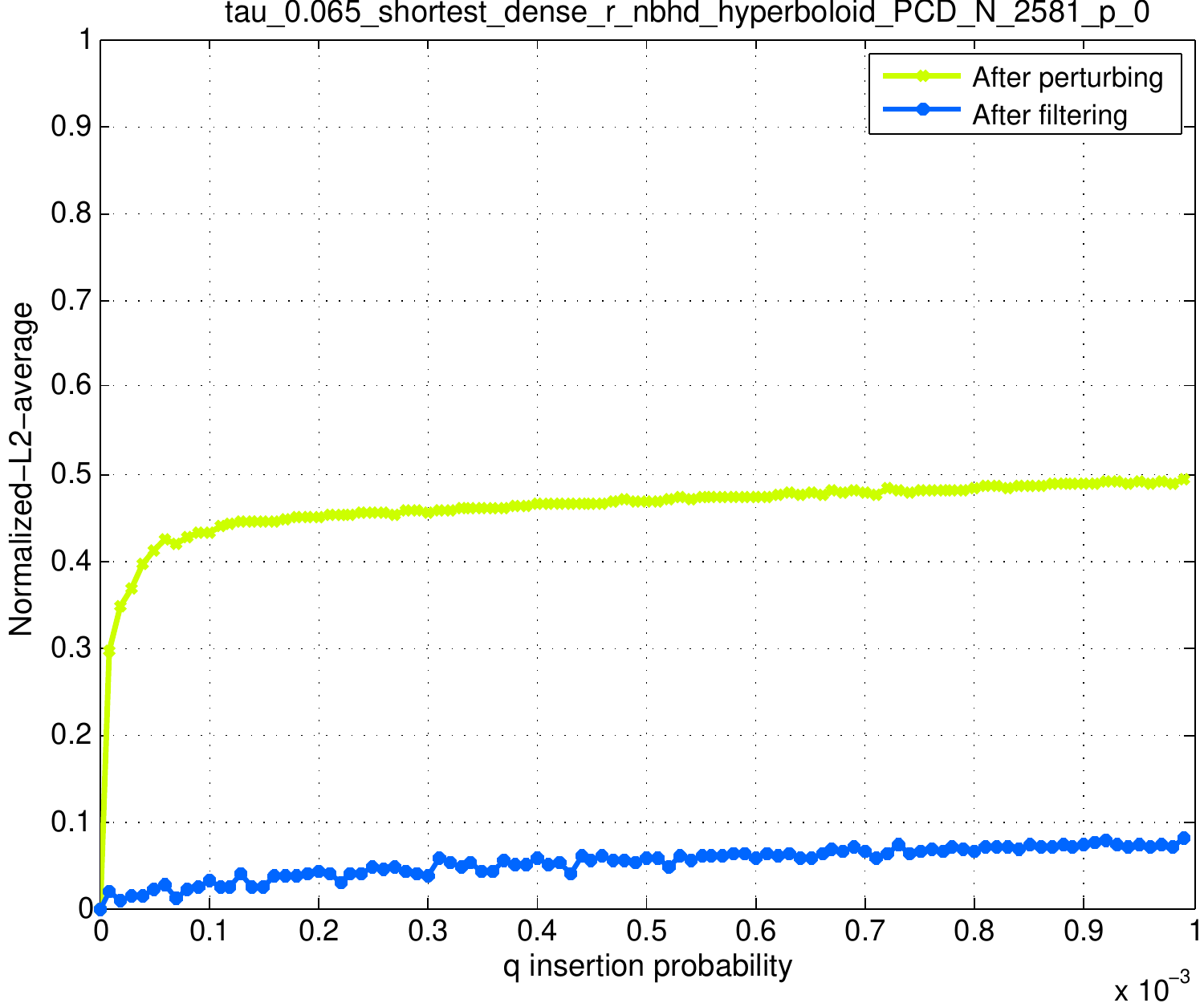}\\
~~~(a) & ~~(b)
\end{tabular}
\caption{(a) ``dense'' hyperboloid $R_{2approx}$. (b) ``dense'' hyperboloid Normalized $L_2$-average error. \label{fig:hyperboloiddense}}
\end{figure}
Figure \ref{fig:hyperboloidsparse} (b)(c) is the result when we apply Jaccard filtering to the ``sparse'' hyperboloid data ($\#$nodes: 2581, $\#$edges: 38321). As we can see, randomly inserting edges distorts the shortest path metrics (with low $2$-approximation rate and high normalized $L_2$ error for $G_q$s). However, our Jaccard-index filtering process restores the metric not only w.r.t $2$-approximation rate (which is predicted by our theoretical results), but also w.r.t normalized $L_2$ error. The plots for the ``dense'' hyperboloid data ($\#$nodes: 2581, $\#$edges: 208290) are shown in Figure \ref{fig:hyperboloiddense}, where we observe similar improvements in error rates. 
Note that the curve for $R_{2approx}$ of the perturbed (but un-filtered) graphs $G_q$ decreases faster with increasing $q$ for the sparse case compared to the dense case; while the curve for the normalized $L_2$ error increases also faster for the sparse case. This fits the intuition that sparse graphs are more sensitive to \myER-type perturbation w.r.t shortest path distance. 

\begin{figure}[htbp]
\centering
\begin{tabular}{cc} 
\includegraphics[width= .35\textwidth]{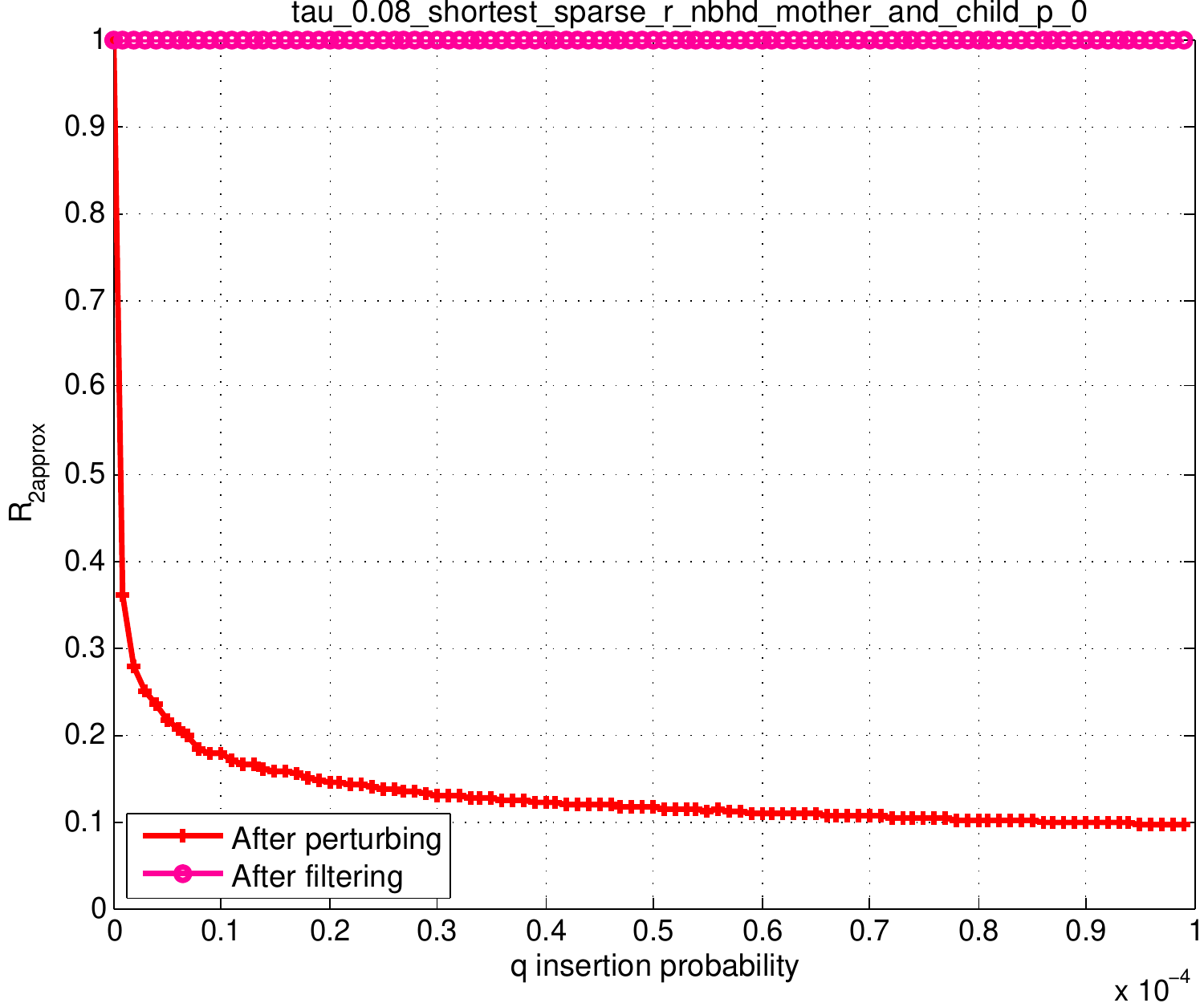} & 
\includegraphics[width= .35\textwidth]{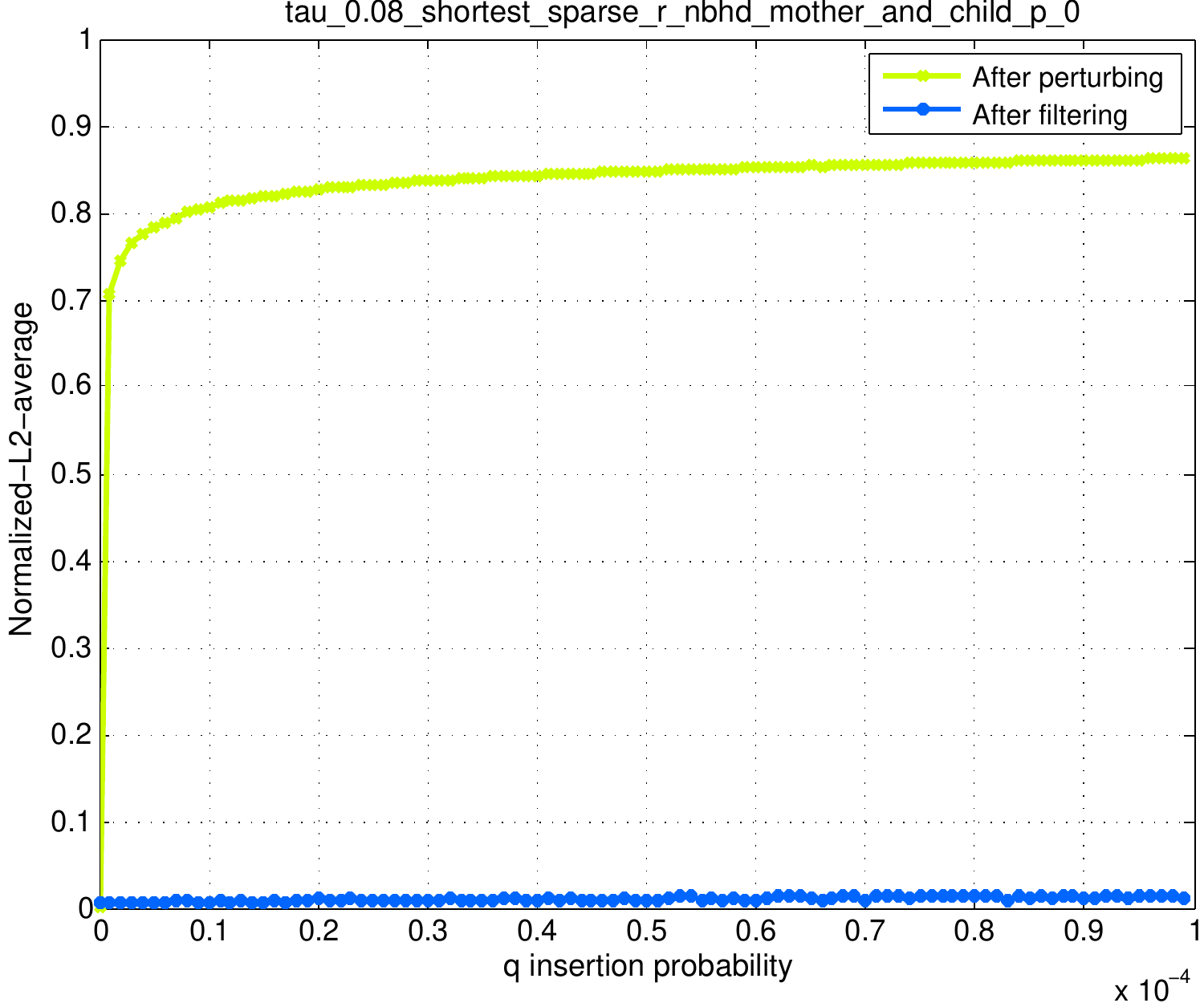}\\
~~~(a) & ~~(b) \\[6pt]
\includegraphics[width= .35\textwidth]{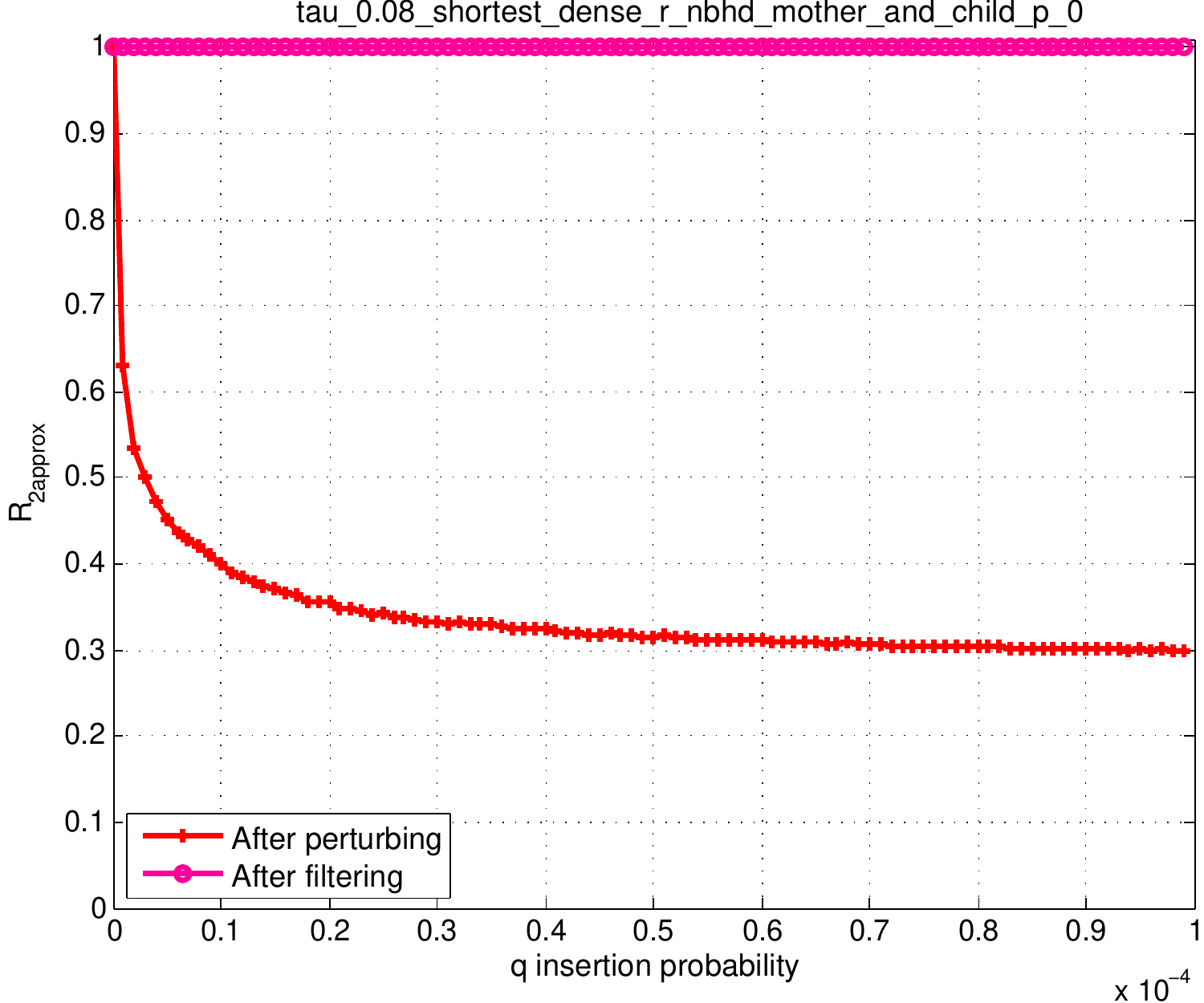} & 
\includegraphics[width= .35\textwidth]{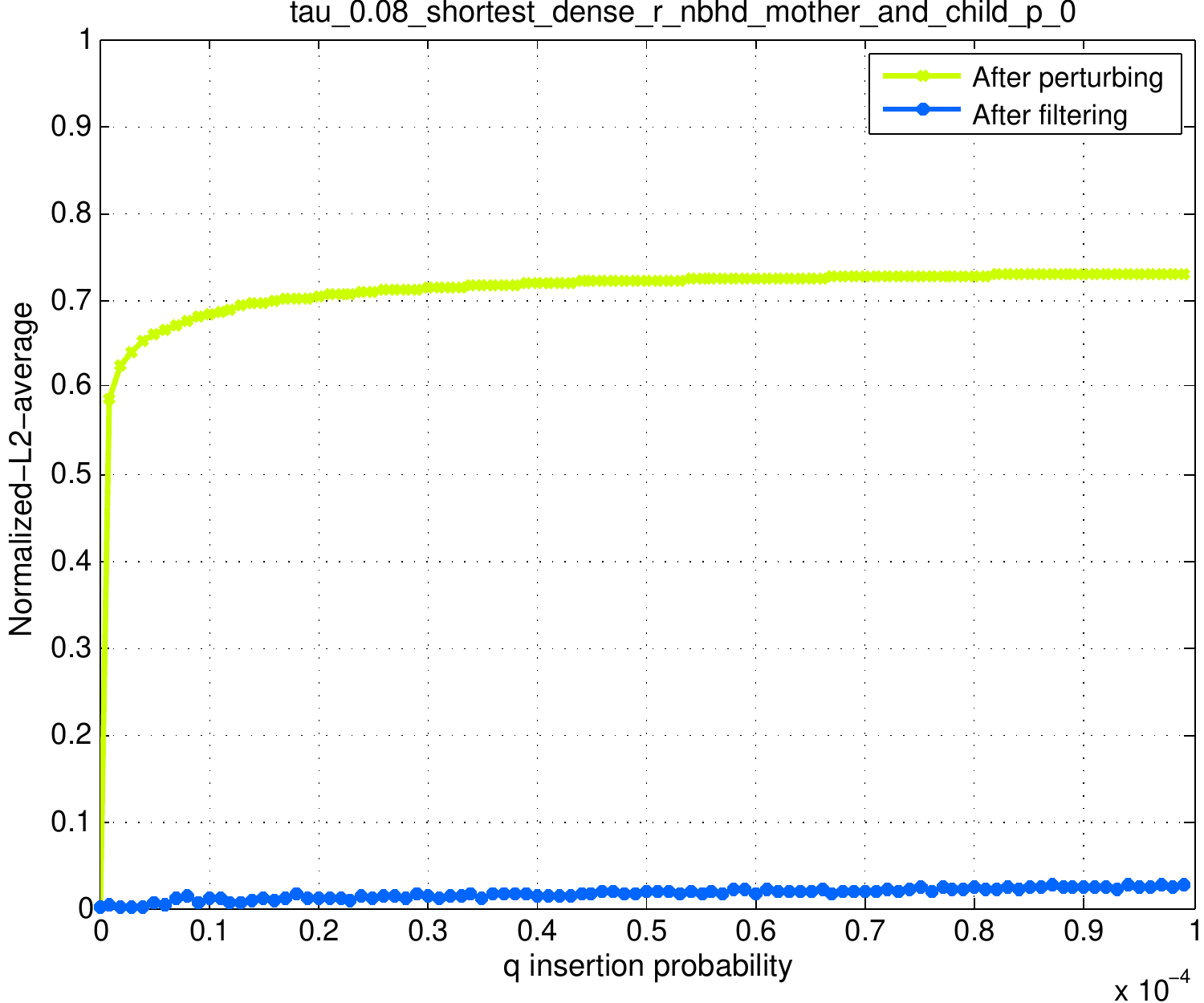}\\
~~~(c) & ~~(d)
\end{tabular}
\vspace*{-0.15in}
\caption{(a) ``sparse'' mother-child model $R_{2approx}$. (b) ``sparse'' mother-child model Normalized $L_2$-average error. 
(c) ``dense'' mother-child model $R_{2approx}$. (d) ``dense'' mother-child model Normalized $L_2$-average error. \label{fig:mothersparse}}
\end{figure}

We perform the same experiments to the mother-child model. Figure \ref{fig:mothersparse} (a) and (b) shows the results for the ``sparse'' mother-child model ($\#$nodes: 23390, $\#$edges: 553797); while the results for the ``dense'' mother-child data ($\#$nodes: 23390, $\#$edges: 3428141) are in (c) and (d). Similar behaviors are observed. 

\begin{figure}[htbp]
\centering
\begin{tabular}{cc}
\includegraphics[width= .35\textwidth]{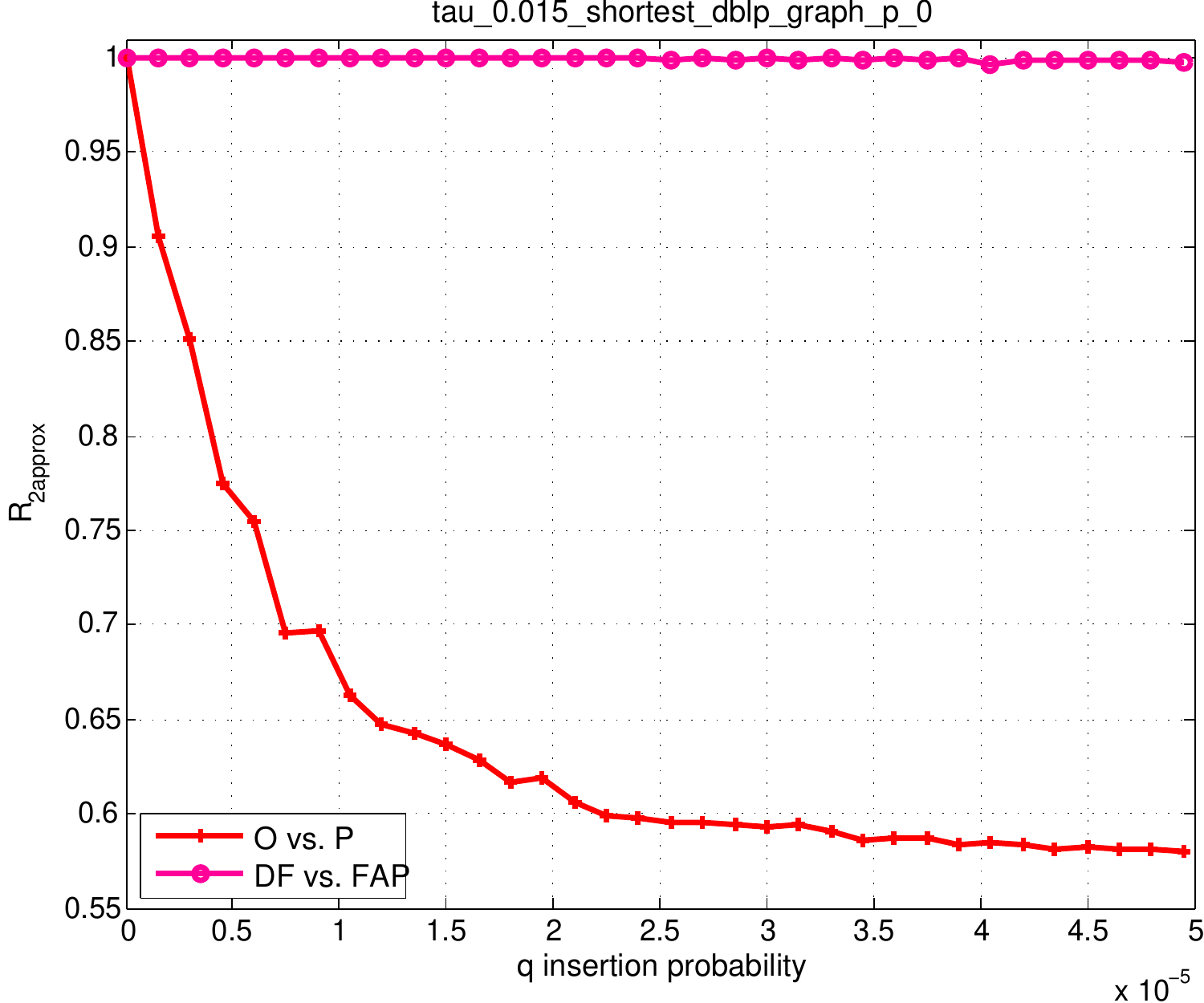} &
\includegraphics[width= .35\textwidth]{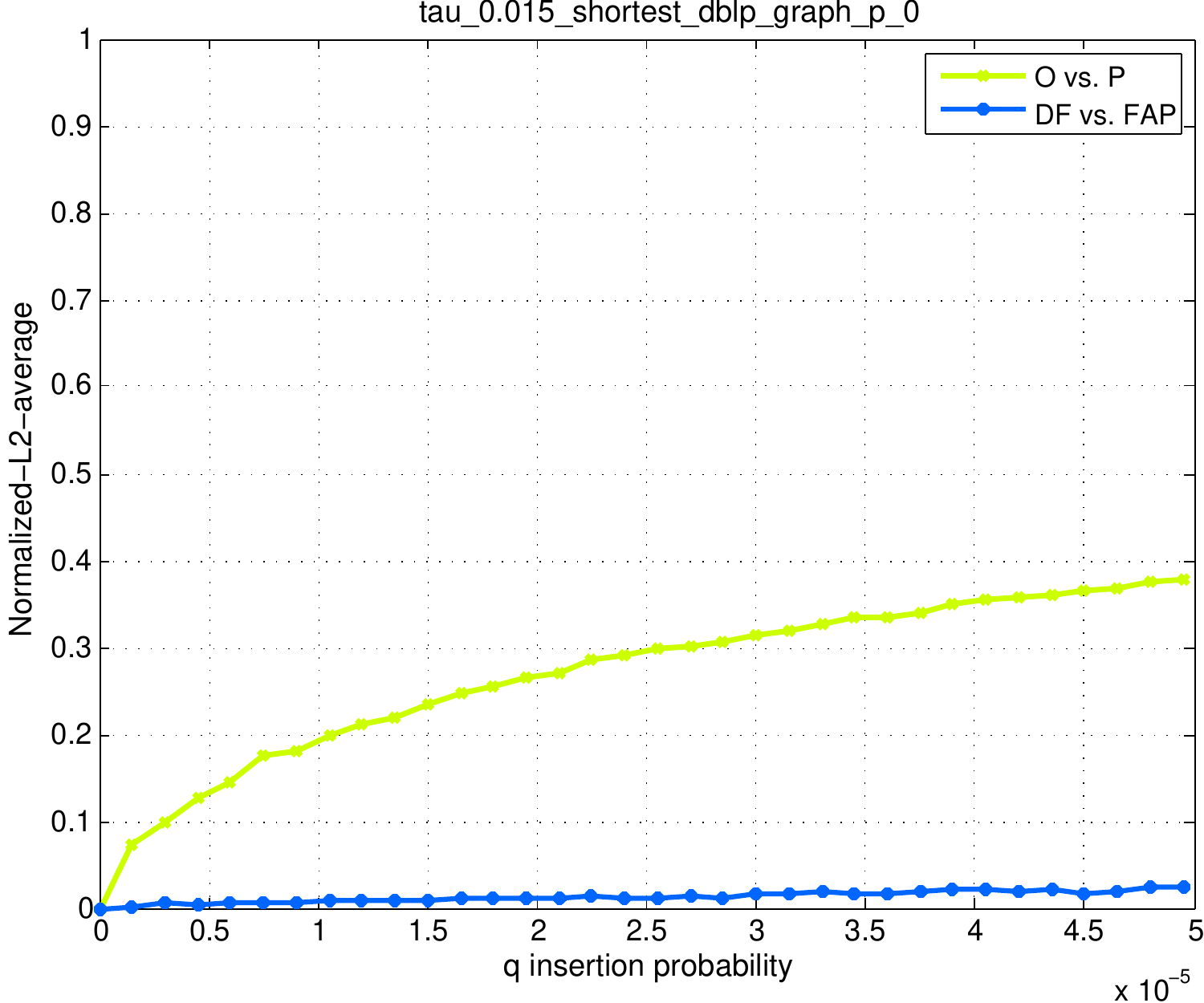} \\
~~~(a) (co-authorship network) $R_{2approx}$. & ~~(b) (co-authorship network) Normalized $L_2$\\
\includegraphics[width= .35\textwidth]{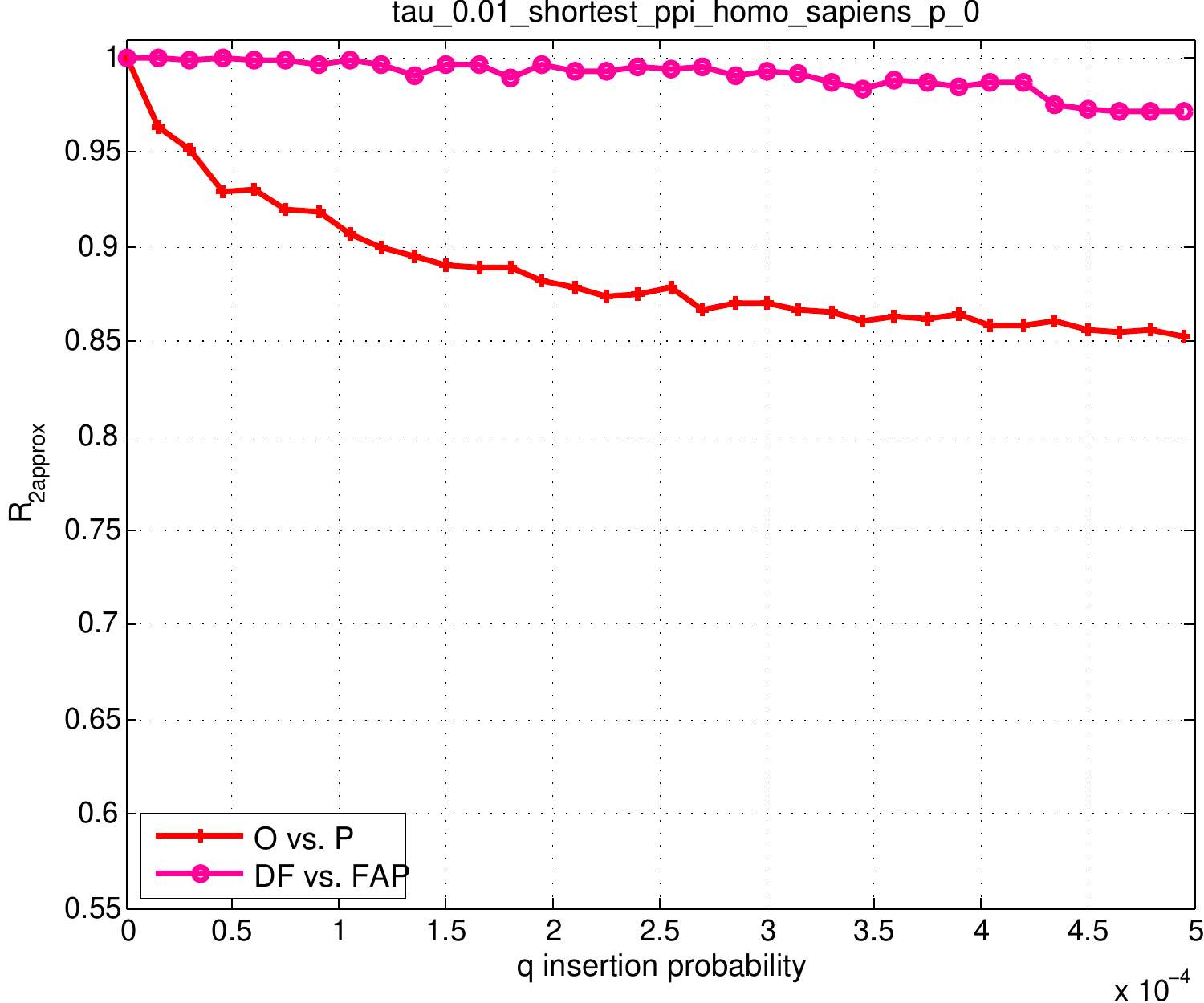} &
\includegraphics[width= .35\textwidth]{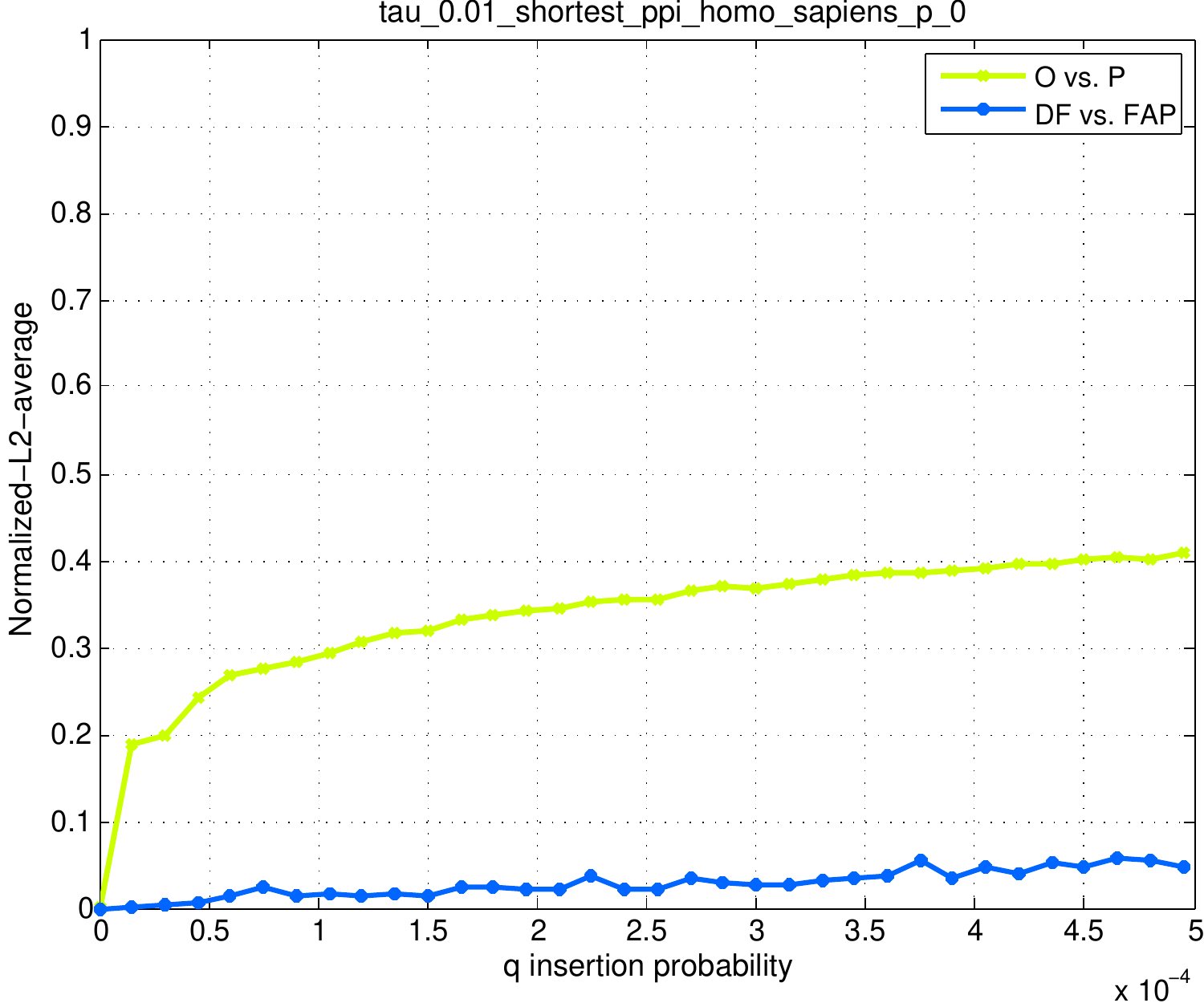} \\
~~~(c) (PPI network) $R_{2approx}$. & ~~(d) (PPI network) Normalized $L_2$\\
\end{tabular}
\vspace*{-0.15in}
\caption{(a)(c) 2-approximation rate $R_{2approx}$ as $q$ (x-axis) increases, with top curve (better rate) for graphs $G^\tau_{q}$ after $\tau$-Jaccard filtering (DF=directly filtered vs. FAP=filtered after perturbed), and bottom one for $G_q$ after only $q$-insertion (O=observed vs. P=perturbed). (b)(d) Normalized $L_2$ error (O vs. P is normalized by O and DF vs. FAP is normalized by DF), with being the one without filtering, and the bottom curve (lower error) for graphs after Jaccard filtering. 
 \label{fig:dblp_graph}}
\end{figure}

\myparagraph{Real networks without ground truth.}
For a given real network $G$, we can consider this as an observed graph. However, we do not know how this network is generated and there is no ground truth graph $\tG$. 
Nevertheless, we carry out the following experiments to indirectly infer the effectiveness of Jaccard-filtering. 
Specifically, given $G$, we gradually add random $(p=0,q)$-perturbation to it, and compare the shortest path metric $D_{G_{q}}$ of the perturbed graph $G_q$ with the metric $D_G$ of input network $G$; $q$ is for the insertion probability. 
Next, we perform $\tau$-Jaccard filtering for all these graphs $G$ and $G_{q}$s to obtain $G^\tau$ and $G^\tau_{q}$ respectively, and then compare the shortest path metric  $D_{G^\tau_q}$ for filtered graphs $G^\tau_q$ with $D_{G^\tau}$ of $G^\tau$. 

See Figure \ref{fig:dblp_graph} (a)(b), where the input network is a network representing co-authorship extracted from papers published at 28 major computer science conferences \cite{LouTang:13WWW} 
($\#$nodes: 53442, $\#$edges: 127969). Note that the normalized $L_{2}$ error is also reduced by Jaccard filtering. 
The results for a protein-protein interaction network \cite{joshi2005reactome} ($\#$nodes: 6327, $\#$edges: 147547) are also given in Figure \ref{fig:dblp_graph} (c)(d). 

\section{Concluding remarks}

In this paper we study how to recover the shortest path metric of a true graph $G^*$ from an observed graph $G$, when $G^*$ is assumed to be some proximity graph of a hidden domain $X$, while $G$ is generated from $G^*$ with random \myER{}-type perturbations. 

Our paper represents one step towards unraveling the structure of the space where data are sampled from. There are many interesting problems along this direction,  including how to generalize our network model to better model real networks. We describe one direction here: Our current work recovers the shortest path metric of the hidden graph $G$. However, there are other common metrics induced from $G$, such as the diffusion distance metric. In fact, for dense random graphs, say graphs generated from a graphon \cite{eldridge2016graphons} (including stochastic block models), the spectral structure of such random graphs are stable. This may imply that diffusion distances could also be stable against random perturbations even without any filtering process. Note that such graphs have $\Theta(n^2)$ number edges asymptotically. 
However, for sparse graphs (which our model could generate), empirically we observe that diffusion distances are not stable under random perturbations. It would be interesting to see whether the Jaccard filtering process (or other filtering procedure) could recover diffusion distances with theoretical guarantees. (Interestingly, we have observed that empirically, Jaccard filtering can recover diffusion distance as well in our experiments.) 

Finally, it would be interesting to explore whether the analysis and ideas for network models from our paper could be used to create a practical wormhole detector in wireless networks, akin to Ban et al's local connectivity tests based on [$\alpha, \beta$]-rings \cite{BSG11}. 

\paragraph{Acknowledgement.} 
The authors thank Samory Kpotufe for the pointer to the local version of $L$-doubling measure. This work is in part supported by National Science Foundation under grants IIS-1550757 and CCF-1618247. SP and DS would like to acknowledge NSF grant \#DMS: 1418265 for partially supporting this work. Any opinions, findings, and conclusions or recommendations expressed in this material are those of the author(s) and do not necessarily reflect the views of the National Science Foundation.

\bibliography{general}

\newpage

\appendix
\section{Relation between metric structures for $\tGr$ and for $\X$}
\label{appendix:thm:metricapprox}

We now prove Theorem \ref{thm:metricapprox} here. 

First, we will argue that $V_n$ forms a dense sampling of the compact space $\X = (X, d_X)$. 
We will then show that under the dense sampling condition, the shortest paths between points in $V_n$ with respect to the input metric $d_X$ is approximated by $d_\tG$ scaled by $r$ as claimed. 

We will start with introducing the concept of $\eps$-sampling. 
\begin{definition}
A finite set of points $P \subset X$ is an $\eps$-sample of $(X, d_X)$ if for any $x \in X$, $d(x, P) \le \eps$ where $d(x,P) = \min_{p \in P} d_X(x, p)$. That is, for any $x\in X$ there is a sample point from $P$ within $\eps$ geodesic distance away from $x$. 
\end{definition}

Now let $\mathcal N_\eps$ denote the \emph{$\eps$-covering number of $X$}, which is the minimum number of closed geodesic balls centered in $X$ of radius $\eps$ needed to cover $X$; denote by $\mathcal B_\eps$ such a collection of geodesic balls with cardinality $\mathcal N_\eps$. 
Set $\mathcal{V}_\eps := \min_{B \in \mathcal B_\eps} \mu(B)$, which is strictly positive and finite since $\mu$ is a doubling measure and $\X$ is compact. 
We claim the following, the proof of which is similar to that of Theorem 5.2 of \cite{chazal2013persistence}: 
\begin{claim}\label{claim:randomepssample}
Let $V_n$ be a set of $n$ points sampled from $(X, d_X)$ w.r.t. $\mu$ in i.i.d. fashion. 
Then $V_n$ forms an $\eps$-sample of $X$ with probability at least $1 - \mathcal{N}_{\eps/2}\cdot e^{-n  \mathcal V_{\eps/2}}$.
\end{claim} 
\begin{proof}
Consider a covering set of geodesic balls $\mathcal B_{\eps/2} = \{B_1, \ldots, B_m\}$ with smallest cardinality $m = \mathcal N_{\eps/2}$. 
For each $i\in [1, m]$, let $E_i$ denote the event that $V_n \cap B_i = \emptyset$. 
Since points in $V_n$ are sampled i.i.d. from $\mu$, we have that 
$$\myprob[E_i] = (1 - \mu(B_i))^n \le (1 - \mathcal V_{\eps/2})^n \le e^{-n \cdot \mathcal V_{\eps/2}}, $$
where the last inequality follows from $1 - x \le e^{-x}$. 
On the other hand, it is easy to see that if $V_n \cap B_i \neq \emptyset$ for all $i\in [1,m]$, then $V_n$ must be an $\eps$-sample of $X$. 
It follows from this and the union bound that
$$\myprob[V_n \text{ is not an } \eps\text{-sample of } X] \le \sum_{i \in [1, m]} \myprob[E_i] \le \mathcal N_{\eps/2} \cdot e^{-n \mathcal{V}_{\eps/2}}. $$ 
The claim then follows. 
\end{proof} 

\begin{claim}\label{claim:epssample-metricbnd}
Suppose that $V_n$ is an $\eps$-sample of $(X,d_X)$ with $\eps < r/2$. Then for any $u, v\in V_n$, 
\begin{equation}\label{eqn:eps-metricbnd}
d_X(u, v) \le r \cdot d_\tG(u, v) \le \frac{r}{r-2\eps} d_X(u, v) + r. 
\end{equation}
\end{claim}
\begin{proof}
Let $\pi \subset X$ be a shortest path (geodesic) from $u$ to $v$ in $X$ with length being $\length(\pi) = d_X(u,v)$.  
Let $v_0 = u, v_1, \ldots, v_k= v$ be the set of vertices in a shortest path $\hat \pi$ from $u$ to $v$ in the $r$-neighborhood graph $\tG$, so $k=d_\tG(u,v)$. 
Now based on $\hat \pi$, we construct the path $\pi' \subset X$ consisting of $k$ pieces, where the $i$th piece is simply the shortest path connecting $v_{i-1}$ to $v_{i}$ in $X$. 
Since $(v_{i-1}, v_{i})$ is an edge in $\tG$, the geodesic distance between them is at most $r$, so we have $\length(\pi') \le r k = r \cdot d_\tG(u,v)$, which implies that 
$$r \cdot d_\tG(u,v) \ge \length(\pi') \ge \length(\pi) = d_X(u,v).$$ 
Hence the left inequality of \eqref{eqn:eps-metricbnd} holds. 

What remains is to bound $r \cdot d_\tG(u,v)$ from above, showing that it cannot be too large compared to $d_X(u,v)$ as well. 

To this end, consider breaking the shortest path $\pi$ (connecting $u$ to $v$ in $X$) at a set of points $p_0 = u, p_1, \ldots, p_s =v \subset \pi$ so that for each $i\in [0, s-2]$, $\length(\pi(p_i, p_{i+1})) = r - 2\eps$ while $\length(\pi(p_{s-1}, p_s)) = d_X(p_{s-1},p_s) \le r - 2\eps$. 
(Note that it is possible that $s = 1$.) 
We then have that 
\begin{equation}\label{eqn:sbnd}
d_X(u,p_{s-1}) = (s-1) \cdot (r - 2\eps), ~\text{implying that}~s - 1 \le \frac{d_X(u,v)}{r - 2\eps}. 
\end{equation} 
Since $V_n$ is an $\eps$-sample of $X$, for each $p_i$, $i\in [1, s-1]$, there exists a  point $u_i \in V_n$ within $\eps$ geodesic distance to $p_i$. Set $u_0 = p_0 = u$ and $u_s = p_s = v$. It then follows from the triangle inequality that 
$$d_X(u_i, u_{i+1}) \le d_X(u_i, p_i) + d_X(p_i, p_{i+1}) + d_X(p_{i+1}, u_{i+1}) \le r,~~\text{for any}~i\in [0, s-1].$$ 
Hence for each $i\in [0, s-1]$, either $u_i = u_{i+1}$ or $(u_i, u_{i+1})$ is an edge in $\tG$. Thus $d_\tG(u, v) \le s$, and combining with \eqref{eqn:sbnd}, the second inequality of \eqref{eqn:eps-metricbnd} follows. 
This proves Claim \ref{claim:epssample-metricbnd}. 
\end{proof}

We now put everything together to prove the theorem. 
By Claim \ref{claim:randomepssample}, for each $\eps$, there exists a sufficiently large integer $n_\eps$ such that for $n\ge n_\eps$, with probability at least $1-e^{-\Omega(n)}$, $V_n$ is an  $\eps$-sample of $X$. Let $\Delta$ denote the diameter of $(X, d_X)$. By Claim \ref{claim:epssample-metricbnd}, if $\eps>0$ is sufficiently small and $n\ge n_\eps$, then with probability at least $1-e^{-\Omega(n)}$, for all $u,v \in V_n$,
\begin{align*}
 | r \cdot d_\tG(u,v) - d_X(u,v) | \le r + \frac{2\eps}{r-2\eps} d_X(u,v) \le r + \frac{3\eps \Delta}{r},
\end{align*}
so $\| r \cdot d_\tG - d_X |_{V_n} \|_\infty \le r + \frac{3\eps \Delta}{r}$.
The second term $\frac{3\eps\Delta}{r}$ tends to zero as $\eps$ tends to $0$.  Since the exceptional probabilities, $e^{-\Omega(n)}$, are summable in $n$ (for fixed $\eps$), the Borel--Cantelli lemma implies that $\limsup_{n\to\infty} \| r \cdot d_\tG - d_X |_{V_n} \|_\infty \le r + \frac{3\eps \Delta}{r}$ almost surely.
Sending $\eps\to 0$ along a countable sequence, we have $\limsup_{n\to\infty} \| r \cdot d_\tG - d_X |_{V_n} \|_\infty \le r$ almost surely.  Theorem \ref{thm:metricapprox} then follows.

\paragraph*{Remark.}
Note that Theorem \ref{thm:metricapprox} only provides an upper bound on the metric difference. This is in some sense necessary, as the graph $\tG$ is an unweighted graph. Hence we cannot differentiate distances from $X$ smaller than $r$.

\end{document}